\documentclass[11pt]{article}
\usepackage{algorithm}
\usepackage{latexsym}
\usepackage[xdvi,dvips]{graphics}
\usepackage{epsfig}
\usepackage{enumerate}

\graphicspath{{figures/}{./}}
\def\abstract{    \if@twocolumn
      \small\it Abstract\/\bf---$\!$    \else
      \begin{small}{\bf Abstract:}
    \fi}
\def\endabstract{\vspace{0.6em}\par\if@twocolumn\else\end{small}\fi
    \normalsize\rm}
\def\qed{\hbox{${\vcenter{\vbox{                                \hrule height
0.4pt\hbox{\vrule width 0.4pt height 6pt
   \kern5pt\vrule width 0.4pt}\hrule height 0.4pt}}}$}}
   
\newenvironment{proof}{{\bf Proof. } }{{\hfill $\Box$}\vspace{.5pc}}

\newtheorem{theorem}{Theorem}[section]
\newtheorem{definition}{Definition}[section]
\newtheorem{lemma}{Lemma}[section]

\newtheorem{corollary}{Corollary}[section]

\newcommand{\PR}{\mathcal P}
\newcommand{\ie}{{\em i.e.}, }
\newcommand{\eg} {{\em e.g.}, }
\newcommand{\mod}{\mbox{ {\tt mod} }}
\newcommand{\Rsgl}{${\mathcal R}_{\mbox{{\tt sgl}}}$}
\newcommand{\Rout}{${\mathcal R}_{\mbox{{\tt out}}}$}
\newcommand{\Rin}{${\mathcal R}_{\mbox{{\tt in}}}$}
\newcommand{\Rswp}{${\mathcal R}_{\mbox{{\tt swp}}}$}
\newcommand{\SP}{${\mathcal R}_{\mbox{{\tt sp}}}$}
\newcommand{\T}[2]{{\mathcal T}{#1}({#2})}
\newcommand{\Tsgl}[1]{{\mathcal T}_{\mbox{{\tt sgl}}}({#1})}
\newcommand{\Tdbl}[1]{{\mathcal T}_{\mbox{{\tt dbl}}}({#1})}

\newcommand{\comment}[1]{\hspace{.1in}// {\it #1}}

\begin{document}

\title{
Ring Exploration with Oblivious Myopic Robots}
\author{Ajoy K. Datta$^{1}$, Anissa Lamani$^{2}$, Lawrence L. Larmore$^{1}$, and Franck Petit$^{3}$\\
$^1$ School of Computer Science, University of Nevada Las Vegas, USA\\
$^2$ MIS, Universit\'{e} de Picardie Jules Verne Amiens, France\\
$^3$ LIP6, INRIA, CNRS, UPMC Sorbonne Universities, France
}

\date{
}

\maketitle

\begin{abstract}
The exploration problem in the discrete universe, using identical oblivious asynchronous robots without direct communication, has been well investigated.  These robots have sensors that allow them to see their environment and move accordingly.  However, the previous work on this problem assume that robots have an unlimited visibility, that is, they can see the position of all the other robots.  In this paper, we consider deterministic exploration in an anonymous, unoriented ring using asynchronous, oblivious, and myopic robots.  By myopic, we mean that the robots have only a limited visibility. We study the computational limits imposed by such robots and we show that under some conditions the exploration problem can still be solved.
We study the cases where the robots visibility is limited to $1$, $2$, and $3$ neighboring nodes, respectively. \\

\end{abstract}


\section{Introduction}
\label{sec:intro}

There has been recent research on systems of autonomous mobile entities
(hereafter referred to as \emph{robots}) 
that must collaborate to accomplish a collective task. 
Possible applications for such multi-robot systems include
environmental monitoring, large-scale construction, 
mapping, urban search and rescue, surface cleaning, risky area surrounding or surveillance, 
exploration of unknown environments, 
and other tasks in environments where instead of humans, robots are used. 
For many of these applications, the larger the number of robots is, the easier the implementation
is. 

However, the ability of a team of robots to succeed in accomplishing the assigned task greatly depends 
on the capabilities that the robots possess, namely, their sensing capabilities. Clearly, the type of viewing 
device has a great impact on the knowledge that the robots have of their environment. For example,
if the robots have access to a global localization system (\eg GPS, egocentric zone-based RFID technology), 
then their viewing capabilities are \emph{a priori} unlimited.
By contrast, endowed with a camera or a sonar, vision capabilities are limited to a certain distance.
Obviously, the stronger the device capabilities are, the easier the problem is solved.

In order to satisfy technological or budget constraints, it may be important
to minimize both aforementioned parameters, \ie ($i$) the number of robots and ($ii$)
capacities (or equipment) required to accomplish a given task.

In this paper, we consider both parameters for the \emph{exploration} problem. 
Exploration is a basic building block for many applications.  For instance, mapping of an unknown area
requires that the robots (collectively) explore the whole area.  Similarly, to search and rescue people after 
a disaster, the team of robots potentially has to explore the whole area.
The so called ``area'' is often considered to be either the {\em continuous} Euclidean space (possibly with obstacles and objects) or a 
{\em discrete} space.  In the latter case, space is partitioned into a finite number of locations represented
by a graph, where nodes represent indivisible locations that can be sensed by the robots, and
where edges represent the possibility for a robot to move from one
location to the other, {\em e.g.}, a building, a town, a factory, a mine, and more generally, zoned areas.
In a discrete environment, exploration requires that 
starting from a configuration where no two robots occupy the same node, 
every node to be visited by at least one
robot, with the additional constraint that all robots eventually
stop moving.

We assume robots having weak capacities:
they are \emph{uniform} --- meaning that all robots 
follow the same protocol ---, \emph{anonymous} --- meaning that
no robot can distinguish any two other robots,
\emph{oblivious} --- they have no memory of any past behavior of themselves
or any other robot, ---, 
and \emph{disoriented} --- they have no labeling of direction.
Furthermore, the robots have no (direct) means of
communicating with each other, \emph{i.e.},
they are unable to communicate together. However, robots are endowed with 
visibility sensors enabling to see robots located on nodes.

In this paper, we add another constraint: \emph{myopia}. 
A myopic robot has limited visibility, \emph{i.e.},
it cannot see the nodes located beyond a certain fixed distance $\phi$.
The stronger the myopia is, the smaller $\phi$ is.  In other words,
we consider that the strongest myopia corresponds to $\phi = 1$, when 
a robot can only see robots located at its own and at neighboring nodes. 
If $\phi = 2$, then a robot can see robots corresponding to $\phi =1$ and
the neighbors of its neighboring nodes.  And so on. 
Note that the weaker myopia corresponds to $\phi = \lceil \frac{D}{2} \rceil - 1$, 
$D$ being the diameter of the graph.  Infinite visibility (\ie each robot is able 
to see the whole graph) is simply denoted by $\phi = \infty$.

We study the impact of myopia strength (\ie the size of the visibility radius) on
the problem of exploration.  
This assumption is clearly motivated by limiting the vision capacities that each robot is required to have. 
As a matter of fact, beyond technological or budget constraints, it is more realistic to assume 
robots endowed with vision equipments able to sense their local (w.r.t. $\phi$) environment than
the global universe.  Furthermore, solutions that work assuming the strongest assumptions also work assuming weaker assumptions.  
In our case, any distributed and deterministic
algorithm that solves the exploration problem with robots having severe myopia ($\phi =1$), also works with robots with 
lower myopia ($\phi > 1$), and even no myopia ($\phi = \infty$).

\paragraph{Related Work.}

Most of the literature on coordinated distributed robots
assumes that the robots move in a \emph{continuous} 
two-dimensional Euclidean space,
and use visual sensors with perfect accuracy,
permitting them to locate other robots with 
infinite precision, \emph{e.g.},\cite{%
AOSY99j,%
DFSY10c,%
DLP08j,%
FPSW05j,%
FPSW08j,%
SY99j%
}. 
In each of these papers,
other than~\cite{AOSY99j,FPSW05j},
the authors assume that each robot is able to see all other robots
in the plane.
In~\cite{AOSY99j}, the authors give an algorithm for
myopic robots with limited visibility,
converging toward a single point that is not known in advance.
The convergence problem is similar to   
the gathering problem where the robots
must meet in a single location in finite time.  
In~\cite{FPSW05j},
the authors present a gathering algorithm for myopic robots in the plane,
which requires that the robots agree on a common coordinate system.

In the discrete model,
{\em gathering} and {\em exploration} are the two main problems
that have been 
investigated so far~\emph{e.g.},~\cite{GuilbaultP11,IIKO10c,KKN08c,KMP08j}
for the gathering problem 
and~\cite{DPT09c,FIPS07c,FIPS10j,LGT10c} for the exploration problem. 
In~\cite{KKN08c},
the authors prove that the gathering problem is not feasible in some symmetric 
configurations and propose a protocol based on
breaking the symmetry of the system. 
By contrast, in~\cite{KMP08j},
the authors propose a gathering protocol that exploits this symmetry
for a large number of robots ($k>18$) closing the
open problem of characterizing symmetric situations on 
the ring that admit gathering solutions.  In~\cite{IIKO10c,KameiLOT11},
the authors achieve similar results assuming weaker 
robot capabilities: robots may not be
able to count the number of robots on the same node.
%
In \cite{GuilbaultP11}, the authors studied the gathering problem considering robots having only a local visibility
{\em ie,} they can only see robots located at its own and at adjacent nodes, \ie $\phi$ is assumed to 
be equal to $1$.  Under this assumption, the authors provide a complete solution of the gathering problem for 
regular bipartite graphs.  
They first characterize the class of initial configurations allowing to solve the
gathering problem on such graphs, and they propose a gathering algorithm assuming that the system starts from 
a configuration in this class.

In~\cite{FIPS10j}, it is shown that, in general, $\Omega(n)$ robots
are necessary to explore a tree network of $n$ nodes. 
In~\cite{FIPS07c}, it is proved that no deterministic exploration
is possible on 
a ring when the number of robots $k$ divides the number of nodes $n$.
In the same paper, the authors proposed a deterministic 
algorithm that solves the problem using at least $17$ robots,
provided that $n$ and $k$ are co-prime. 
In~\cite{DPT09c}, it is shown that no protocol (probabilistic or deterministic)
can explore a ring with fewer than four robots. 
In the same paper, the authors give a probabilistic algorithm
that solves the problem on a ring of size $n > 8$ 
that is optimal for the number of robots. 
Finally, in \cite{LGT10c}, the authors reduce the gap in the
deterministic case between a large upper bound 
($k\geq 17$) and a small lower bound ($k>3$)
by showing that $5$ robots are necessary and sufficient in the case 
that the size of the ring is even, and that $5$ robots are sufficient
when the size of the ring is odd.

\paragraph{Contribution.}
To the best of our knowledge, all previous results for discrete versions of the exploration problem assume 
unlimited visibility ($\phi = \infty$), \ie the whole graph is seen by each robot. In this paper, we consider 
deterministic algorithms for ring networks of $n$ nodes that uses $k$ myopic
robots for ring networks. 
Following the same arguments as in~\cite{GuilbaultP11}, it should be emphasized that exploration 
would not be solvable starting from any configuration but from configurations where each robot is at distance 
at most $\phi$ of another one---at least when $\phi = 1$. 

Our contribution is threefold.  First, we tackle the case where robots have a visibility $\phi =1$, \ie they can only
see the state of nodes at distance $1$. In this case, we show that the exploration problem is not solvable in both 
semi-synchronous and asynchronous models, for $n>1$ and $1 \leq k\leq n$.
Next, we show that even in the (fully) synchronous model, the exploration problem 
cannot be solved with less than $5$ robots when $n>6$. We then propose, 
optimal deterministic algorithms in the synchronous model for both cases $3\leq n \leq 6$ and $n>6$.
Second, we propose two deterministic solutions in the asynchronous model when robots have a visibility $\phi =2$
using respectively $7$ (provided $n>k$) and $9$ (provided that $n>k\phi +1$) robots.  
Finally, we show that no exploration is
possible with less than $5$ robots when they have a visibility $\phi =3$ and $n>13$ in both semi-synchronous and
asynchronous model.  We then propose two asynchronous solutions that solve the problem using respectively $5$ and $7$
robots. The former works starting from specific configurations.  
All our solutions work assuming that 
each robot is at distance at most $\phi$ of another one.  Both solutions for $7$ robots with $\phi = 2$ and $5$ robots
with $\phi = 3$ work starting from specific configurations.

\paragraph{Roadmap.}
Section~\ref{sec:prel} presents the system model that we
use throughout the paper. We present our results for the cases $\phi=1$, 
$\phi=2$, and $\phi=3$ in Section~\ref{sec:phi=1}, \ref{sec:phi=2}, and~\ref{sec:phi=3}, 
respectively. 
Section~\ref{sec:conclu} gives some concluding remarks.


\section{Model}
\label{sec:prel}

\paragraph{System.}

We consider systems of autonomous mobile entities called {\em robots} 
moving into a discrete environment modeled by a graph $G=(V,E)$, 
$V$ being a set of nodes representing a set of {\em locations}
(or {\em stations}) where robots are, 
$E$ being a set of edges that represent
bidirectional connections through which robots move from a station to another. 
We assume that the graph $G$ is a \emph{ring} of $n$ nodes, 
$u_0,\dots, u_{n-1}$, {\em i.e.}, $u_i$ is connected
to both $u_{i-1}$ and $u_{i+1}$ --- every computation
over indices is assumed to be 
modulo $n$. The indices $0\ldots i \ldots n-1$
are used for notation purposes only; the nodes are {\em anonymous} 
and the ring is {\em unoriented},
{\em i.e.}, given two neighboring nodes $u$, $v$, there is
no kind of explicit or implicit labeling allowing
to determine whether $u$ is on the right
or on the left of $v$.   
On this ring, $k \leq n$ anonymous robots
$r_0, \ldots, r_j, \ldots, r_{k-1}$
are collaborating together to achieve a given task.
The indices $0\ldots j \ldots k-1$ are used for notation
purposes only, since they are undistinguished.
Additionally, the robots are {\em oblivious}, \textit{i.e},
they have no memory of their past actions. 
We assume the robots do not communicate in an explicit way.
However, they have the ability to sense their environment.
%
Zero, one, or more robots can be located on a node.
The number of robots located on a node $u_i$ at instant $t$ is called {\em
multiplicity} of $u_i$ and is denoted by $M_{i}(t)$
(or simply $M_i$, if $t$ is understood).
We say a node $u_i$ is
{\em free\/} at instant $t$
if $M_{i}(t)=0$. Conversely, we say
that $u_i$ is occupied at instant $t$
if $M_{i}(t)\ne 0$. 
If $M_{i}(t)>1$ then, we say that there is an $M_i(t)$.{\em tower\/} 
(or simply a {\em tower}) at $u_i$ at instant $t$.

We assume that each robot is equipped with an abstract device called
{\em multiplicity sensor} allowing to measure node multiplicity.  
Multiplicity sensors are assumed to have {\em limited capacities}
(so called, {\em limited visibility}), {\em i.e.}, 
each robot $r_j$ can sense the multiplicity of nodes
that are at most at a fixed distance\footnote{
The distance between two vertices in the ring is the number of edges
in a shortest path connecting them.}
$\phi$ ($\phi>0$) from the node where $r_j$ is located. We assume that $\phi$ is a common value to all the robots.
The ring being unoriented,
no robot is able to give an orientation to its view.  
More precisely,
given a robot $r_j$ located at a node $u_i$, the multiplicity sensor
of $r_j$ outputs a sequence, $s_j$, of $2\phi+1$ integers
$x_{-\phi}, x_{1-\phi}, \ldots , x_{-1}, x_0, x_1, \ldots , x_{\phi-1}, x_{\phi}$
such that: 

\noindent\begin{tabular}{ll}
either & $x_{-\phi} = M_{i-\phi}$, 
         \ldots, 
         $x_0 = M_i$, \ldots,
         $x_{\phi} = M_{i+\phi}$,\\
or     & $x_{-\phi} = M_{i+\phi}$, 
         \ldots,
         $x_0 = M_i$, \ldots,
         $x_{\phi} = M_{i-\phi}$.
\end{tabular}

If the sequence $x_1, \ldots , x_{\phi-1}, x_{\phi}$
is equal to the sequence  $x_{-1}, \ldots ,x_{1-\phi}, x_{-\phi}$, then 
the view of $r_j$ is said to be {\em symmetric}.
Otherwise, it is said to be {\em asymmetric}.  



%
%

\paragraph{Computations.}

Robots operate in three phase cycles: Look, Compute and
Move (L-C-M). During the Look phase, a robot $r_j$ located at $u_i$ 
takes a snapshot of its environment given by the output of its
multiplicity sensors, {\em i,e.}, a sequence $s_j$ of
$2\phi+1$ integers. 
Then, using $s_j$, $r_j$ computes a destination to move to, {\em i.e.},
either $u_{i-1}$, $u_i$, or $u_{i+1}$. 
In the last phase (move phase),
$r_j$ moves to the target destination computed in the previous phase.   

Given an arbitrary orientation of the ring and a node $u_i$,
$\gamma^{+i}(t)$ (resp., $\gamma^{-i}(t)$) denotes the sequence
$\langle M_i(t) M_{i+1}(t) \dots  M_{i+n-1}(t)\rangle$
(resp., $\langle M_i(t) 
\dots  M_{i-(n-1)}(t)\rangle$).
We call the sequences $\gamma^{-i}(t)$ and $\gamma^{+i}(t)$ {\em mirrors\/}
of each other.  Of course, a symmetric sequence is its own mirror.
By convention, we state that the {\em configuration} of the system at
instant $t$ is $\gamma^{0}(t)$. 
Let $\gamma = \langle M_0 M_{1} \dots  M_{n-1}\rangle$ be a configuration.
The configuration $\langle M_i, M_{i+1}, \dots , M_{i+n-1}\rangle$
is obtained by rotating $\gamma$ of $i \in
[0\ldots n-1]$.  
Two configurations $\gamma$ and $\gamma'$ are said to be
{\em undistinguishable} if and 
only if $\gamma'$ is obtained by rotating $\gamma$ or its mirror.
Two configurations that are not
undistinguished are said to be {\em distinguished}.

We call a configuration $\gamma=\langle{M_0 \dots M_{n-1}\rangle}$
\emph{towerless} if $M_i\leq 1$. for all $i$.
A configuration at which no robot that can move we call  \emph{terminal}. 



An {\em inter-distance} $d$ refers to the minimum distance taken among
distances between each pair of
distinct robots. 
Given a configuration $\gamma$, a {\em $d$.block} is any maximal
elementary path in which robots are at distance $d$. 
Each occupied node at the extremity of the $d$.block is called a {\it border}.  
The \emph{size} of a $d$.block is the number
of robots in the $d$.block.
A robot not being into a $d$.block is said to be \emph{isolated}.
A {\em $\phi$.group} is any maximal elementary path in which there is one
robot every node at distance at most $\phi$ of each other. 
In other words, a $\phi$.group is a $d$-block for the particular case
of $d=\phi$.

At each step $t$, a non-empty subset of robots is {\em selected\/} the
\textit{scheduler}, or {\em daemon}.
The scheduler is viewed as an abstract external entity.
We assume a \emph{distributed fair} scheduler. Distributed means that, at
every instant, any non-empty subset of robots can be activated.  Fair
means that every robot is activated infinitely often during a
computation. 
%
We consider three computational models: ($i$) The {\em semi-synchronous} model, ($ii$) the ({\em fully}) {\em synchronous} model, 
and ($iii$) the {\em asynchronous} model. 
In the former model, at every time instant $t$, every robot that is selected
instantaneously executes the full cycle L-C-M 
between $t$ and $t+1$.  (This model is known as the ATOM model~\cite{SY99j}). 
The synchronous model is similar to the semi-synchronous model, except that
the scheduler selects all enabled robots at each step.
In the asynchronous model, cycles L-C-M are performed asynchronously
for each robot, {\em i.e.}, the time between Look, Compute, and Move operations
is finite but unbounded, and is decided by the scheduler for each action of
each robot.  Note that since each robot is assumed to be located at a node,
the model considered in our case can be seen as
\textit{CORDA}~\cite{Pre01} with the following constraint:
the Move operation is {\em atomic}, \textit{i.e}, no robot 
can be located on an edge.
In other words, whenever a robot takes a snapshot of its environment,
the other robots are always seen on nodes, never on edges. 
Since the scheduler is allowed to interleave the operations,
a robot can move according to an outdated view, {\em i.e.},
during the computation phase, some robots may have moved. 

We call \emph{computation} any infinite sequence of configurations
$\gamma_0,$ $\dots,$ $\gamma_t,$ $\gamma_{t+1},$ $\dots$ such that (1)
$\gamma_0$ is a possible initial configuration and (2) for every
instant $t \geq 0$, $\gamma_{t+1}$ is obtained from $\gamma_t$ after
some non-empty set of robots executes an action. Any transition $\gamma_t,
\gamma_{t+1}$ is called a step of the computation. A computation $c$
\emph{terminates} if $c$ contains a terminal configuration.
%

\paragraph{Algorithm.}
Each rule in the algorithm is presented in the following manner: \newline
%
\noindent $<Label>\; <Guard>$ $::$ $<Statement>$. 
%
The guard is a possible sequence $s$ provided by the sensor of a robot
$r_j$:
$ s = x_{i-\phi}, \ldots ,x_{i-1}, (x_i),
x_{i+1}, \ldots , x_{i+\phi}$.
A robot $r_j$ at node $u_i$ is {\em enabled} at time $t$
(or simply  {\em enabled} when it is clear from the context) 
if: \newline
\noindent $s = M_{i-\phi}(t), 
\ldots ,M_{i-1}(t), (M_i),
M_{i+1}(t), \ldots ,
M_{i+\phi}(t)$, or \newline
\noindent 
$s = M_{i+\phi}(t), 
\ldots ,M_{i+1}(t), (M_i),
M_{i-1}(t), \ldots ,
M_{i-\phi}(t)$. 
The corresponding rule $<Label>$ is then also said to be {\em enabled}.
The statement describes the action to be performed by $r_j$.
There are only two possible actions: 
($i$) $\rightarrow$, meaning that $r_j$ moves towards the node $u_{i+1}$, 
($ii$) $\leftarrow$, meaning that $r_j$ moves towards the node $u_{i-1}$.
Note that when the view of $r_j$ is symmetric,
the scheduler chooses the action to be performed. 
In this case, we write statement:  $\leftarrow$ $\vee$ $\rightarrow$.

Character {\tt '?'} in the algorithms means {\em any value}.

\paragraph{Problem Specification.}

%

\begin{definition}[Exploration]
\label{def:exploration}
Let $\PR$ be a \emph{deterministic} protocol designed for a team of
$k$ robots with a given positive
visibility $\phi$, evolving on an $n$-size ring.  
$\PR$ is a {\em deterministic} ({\em terminating}) {\em exploration} protocol 
if and only if 
every computation $c$ of $\PR$ starting from any towerless configuration, 
the $k$ robots collectively explore the ring, {\i.e.},
($i$) $c$ terminates in \emph{finite time}, 
($ii$) Every node is visited by at least one robot during $c$. 
\end{definition}



\begin{theorem}
\label{th:undist}
If $\PR$ is a deterministic exploration protocol,
then for every pair of distinct configurations
$\gamma_i$ and $\gamma_j$ in any execution $c$ of $\PR$,
$\gamma_i$ and $\gamma_j$ are distinguished. 
\end{theorem}
\begin{proof}
Suppose $\PR$ is a deterministic exploration protocol,
and there exists an execution $c$ of $\PR$
and a pair of distinct configurations, $\gamma_i$ and $\gamma_j$ of $c$,
such that $\gamma_i$ and $\gamma_j$ are undistinguished.
Since $\PR$ is a deterministic exploration protocol, $c$ terminates
in finite time. Let $s=\gamma_0 \gamma_1 \ldots
\gamma_{l-1}$ be the smallest prefix of $c$ of length $l$ that contains
distinguishable configurations only.
Recall that we assume that no two consecutive
configurations of $c$ are identical.
Since $s$ is the smallest prefix of $c$ containing only
distinguished configurations
by assumption there exists $\gamma_i$ in $s$
such that $\gamma_i$ and $\gamma_l$ are undistinguished. Such a
$\gamma_l$ exists since we assume that $c$ contains at least two
undistinguished configurations.
So, executing $\PR$, the set of actions
that led the system from $\gamma_i$ to $\gamma_{i+1}$
are also applicable in $\gamma_l$.
Thus, by executing the same set of actions,
there exists a configuration $\gamma_{l+1}$ reachable from $\gamma_l$ 
such that $\gamma_{i+1}$ and $\gamma_{l+1}$ are undistinguished.
Following the same reasoning for every 
$j \in [2..l-1]$,  there exists a configuration $\gamma_{l+j}$
reachable from $\gamma_{l+j-1}$ 
such that $\gamma_{i+j}$ and $\gamma_{l+j}$ are undistinguished.
Therefore, there exists an execution $c$ of $\PR$
leading to the infinite sequence of actions starting for $\gamma_i$.
This contradicts the assumption that $c$ terminates in finite time.
\end{proof}


\section{Visibility $\phi=1$}
\label{sec:phi=1}

In this section, we first prove that no deterministic exploration is possible in the semi-asynchronous 
model when robots
are able to see at distance $1$ only ($\phi=1$). 
The result holds for any $k<n$.  We then show that no deterministic exploration
solves the problem with four robots, even in the (fully) synchronous model when $n>6$. 
The above results are also valid for the asynchronous model~\cite{P02}. 
Next, we provide optimal deterministic algorithms in the 
synchronous model for both cases $3\leq n \leq 6$ and $n>6$.

\subsection{Negative Results}
\label{sub:imp1}


\paragraph{Asynchronous Model.}
Since robots are able to see only at distance $1$,
only the four following rules are possible: 


\begin{tabular}{rl}
{\Rsgl} & 0(1)0 :: $\rightarrow$ $\vee$ $\leftarrow$\\
{\Rout} & 0(1)1 :: $\leftarrow$\\
{\Rin} & 0(1)1 :: $\rightarrow$\\
{\Rswp} & 1(1)1 :: $\rightarrow$ $\vee$ $\leftarrow$
\end{tabular} 
\\

In the following, $\gamma_0$ refers to an initial configuration. Let us first assume that $\gamma_0$ consists of a
single $1$.block of size $k$. 
Note that Rule~{\Rsgl} is not applicable in an initial configuration.  Also, Rules~{\Rout} and~{\Rin}
(Rule~{\Rswp}) implies that $k$ must be greater than $2$ ($3$, respectively).
We first show that no deterministic exploration protocol includes Rules {\Rsgl}, {\Rout}, and {\Rswp} starting from $\gamma_0$.
We now show that Rules~{\Rout} and~{\Rswp} cannot be part of the protocol. 


\begin{lemma}\label{2}
Let $\PR$ be a semi-synchronous protocol for $\phi = 1$ and $2\leq k < n$. If $\PR$ solves the exploration problem, 
then $\PR$ does not include Rule~{\Rout}. 
\end{lemma} 

\begin{proof}
By contradiction, assume that $\PR$ includes Rule~{\Rout}.  Note that Rule~{\Rout} is only enabled on 
robots that are at the border of a $1$.block.  Since there is only one $\phi$-group in $\gamma_0$ (the initial
configuration), there are only two robots that can execute Rule~{\Rout}.  
Let $u_i u_{i+1} \ldots u_{i+k}$ be the $1$.block in $\gamma_0$.
Denote $r_j$ the robot located on node $u_{i+j}$ ($0\leq i < k$). 
%
Without lost of generality, assume that the adversary activates $r_0$ only in $\gamma_0$. 
Let us call the resulting configuration by $T$ (standing for ``{\em Trap}'' configuration).
There are two cases to consider:
\begin{itemize}
\item $k = n-1$.  
In that case, $r_0$ moves to $u_{i+k+1}$ and $T$ includes a $1$.block $u_{i+1} \ldots u_{i+k} u_{i+k+1}$. 
$T$ is undistinguishable from $\gamma_0$. This contradicts Theorem~\ref{th:undist}.


\item $k < n-1$.  
In that case, once $r_0$ moves, it becomes an isolated robot on $u_{i-1}$ in $T$.  Again, there are two cases. 
Assume first that $k=2$.  Then, $T$ includes two isolated robots, $r_0$ and $r_1$.  
Even if $\PR$ includes Rule~{\Rsgl}, by activating $r_0$ and $r_1$ simultaneously,
 $T+1$ is undistinguishable from $T$.  Again, this contradicts Theorem~\ref{th:undist}.  
Assume that $k>2$.  Then, $T$ is the configuration in which 
$u_{i+1}$ is the border occupied by $r_1$ of the $1$.block $u_{i+1} \ldots u_{i+k}$.
Assume that in $T$, the adversary activates $r_1$ that executes Rule~{\Rout}.  
Then, Configuration~$T+1$ includes two $1$.blocks disjoint by one node.  The former includes $r_0$ and $r_1$ 
(on $u_{i-1}$ and $u_i$, respectively).  The latter forms the sequence $r_2\ldots r_k$, located on $u_{i+2} \ldots
u_{i+k}$.
More generally, assume that $\forall i \in [1..k-2]$, in $T+i$, the adversary activates $r_{i+1}$.  The resulting 
configuration $T+k-2$ is undistinguishable from $T$.  A contradiction (Theorem~\ref{th:undist}). 
\end{itemize}
 
\end{proof}

\begin{lemma}\label{3-4}
Let $\PR$ be a semi-synchronous protocol for $\phi = 1$ and $2\leq k < n$. If $\PR$ solves the exploration problem, 
then  Rule~{\Rin} and Rule~{\Rswp} are mutually exclusive with respect to $\PR$. 
\end{lemma} 
\begin{proof}
Assume by contradiction that $\PR$ includes both Rule~{\Rin} and Rule~{\Rswp}.
Note that Rule~{\Rin} (respectively, Rule~{\Rswp}) is enabled only on robots that are located at the
border of the $1$.block (resp., inside the $1$.block).
In the case where $k$ is even, assume that the scheduler activates all the robots. The resulting configuration is undistinguishable from $\gamma_0$.
A contradiction (by Theorem~\ref{th:undist}). In the case where $k$ is odd, suppose that the scheduler activates the robots at the border of the $1$.block and their neighbors (in the case where $k=3$, only one extremity is activated). The resulting configuration is undistinguishable from $\gamma_0$.
A contradiction (by Theorem~\ref{th:undist}).
\end{proof}

\begin{lemma}\label{3}
Let $\PR$ be a semi-synchronous protocol for $\phi = 1$ and $2\leq k < n$. If $\PR$ solves the exploration problem, 
then $\PR$ does not include Rule~{\Rswp}.
\end{lemma} 
\begin{proof}
By contradiction, assume that $\PR$ includes Rule~{\Rswp}.
From Lemmas~\ref{2} and~\ref{3-4}, $\PR$ does not include Rules~{\Rout} nor Rule~{\Rin}.  So, in $\gamma_0$, the 
robots that are inside a $1$.block are the only robots that are able to execute an action in the initial configuration
(and $k$ must be greater than or equal to $3$).  
If $k>3$, then, the adversary activates at each time two neighboring robots such that once they move, they simply exchange their position and the resulting configuration is undistinguishable from $\gamma_0$. A contradiction.\\
If $k=3$ then, when Rule~{\Rswp} is executed, a $2$.tower is created. Rule $0(2)0$:: $\leftarrow$ $\vee$ $\rightarrow$ cannot be enabled since once executed, the scheduler can activate only one robot in the $2$.tower such that the configuration reached is undistinguishable from $\gamma_0$. In the case where {\Rsgl} is enabled, then the isolated robot becomes either ($i$) neighbor of the $2$.tower (the case where $n=4$) or $(ii)$ remains isolated. In the first case $(i)$, if either $0(1)2$:: $\leftarrow$ or $0(2)1$:: $\rightarrow$ is enabled, a $3$.tower is created (Observe that one node of the ring has not been explored). The only rule that can be executed is $0(3)0$:: $\leftarrow$ $\vee$ $\rightarrow$. Suppose that the scheduler activates two robots that move in two opposite directions. The configuration reached is undistinguishable from $\gamma_0$. If $0(2)1$:: $\leftarrow$ is enabled, then suppose that the scheduler activates only one robot. The configuration reached is undistinguishable from $\gamma_0$. In the second case ($ii$) (the isolated robot remains isolated), {\Rsgl} keeps being enabled. Once it it is executed, the isolated robot moves back to its previous position and the configuration reached is indistinguishable from the previous one. 
\end{proof}


\begin{corollary}\label{cor:init}
Let $\PR$ be a semi-synchronous protocol for $\phi = 1$, $2\leq k < n$. If $\PR$ solves the exploration 
problem, then $\PR$ includes Rule~{\Rin} that is the only applicable rule in the initial configuration.
\end{corollary}
\begin{proof}
Directly follows from Lemmas~\ref{2}, \ref{3-4}, and \ref{3}.
\end{proof}

\begin{lemma}\label{lem:kt5}
Let $\PR$ be a semi-synchronous exploration protocol for $\phi = 1$, $2\leq k < n$. Then, 
$k$ must be greater than or equal to $5$.
\end{lemma} 
\begin{proof}
Assume that the adversary activates both robots $r_0$ and $r_1$ at the border of the $1$.block in $\gamma_0$.
By Corollary~\ref{cor:init}, both  $r_0$ and $r_1$ execute Rule~{\Rin} in $\gamma_0$.  

If $k=2$, the system reaches a configuration $\gamma_1$ that is undistinguishable from $\gamma_0$.

If $k=3$, then a $3$.tower is created on the middle node of the $1$.block in $\gamma_1$. Since the ring is not
explored anymore (at least one node is not visited),
$\PR$ must include at least the following rule: $0(3)0::\; \leftarrow \vee \rightarrow$.  
Then, the adversary activates only two of them and an opposite destination node for each of them. 
Then, the system reaches a configuration $\gamma_2$ that is undistinguishable from $\gamma_0$.

If $k=4$, then two neighboring $2$.towers are created on the two middle nodes of the $1$.block in $\gamma_1$.  
There are two cases to consider:
Assume first that $\PR$ includes the rule: $2(2)0::\; \leftarrow$. Then, by activating the same number of robots on each tower,
the system reaches $\gamma_2$ such that $\gamma_1$ and $\gamma_2$ are two undistinguishable configurations.  
Assume that $\PR$ includes the rule: $2(2)0::\; \rightarrow$. Again, by activating only one robot on each tower, 
the system reaches a configuration $\gamma_2$ that is undistinguishable from $\gamma_0$.
\end{proof}


Let us call a {\em symmetric $x$.tower sequence} (an $S^x$-sequence for short),
a sequence of occupied nodes such that the extremities of the sequence contain an $x$.tower. 
Observe that a tower containing at least $2$ robots (by definition), $x$ is greater than
or equal to $2$.  Also, 
since the robots can only see at distance $1$, the tower is only seen by its neighboring robots. 

The following lemma directly follows from Corollary~\ref{cor:init} and Lemma~\ref{lem:kt5}:
\begin{lemma}\label{lem:sur}
For every semi-synchronous exploration protocol $\PR$ for $\phi = 1$, then $5\leq k < n$ and  
there exists some executions of $\PR$ leading to a configuration containing an $S^2$-sequence.  
\end{lemma} 

Consider robots being located at the border of an $1$.block.  
Let $x\geq 1$ be the number of robots located on the border node and the 
following generic rules:

\begin{tabular}{rl}
$\T{\alpha}{x}$ & $0(x)1 ::\; \leftarrow$\\
$\T{\beta}{x}$  & $0(x)1 ::\; \rightarrow$\\
$\T{\gamma}{x}$ & $x(1)1 ::\; \leftarrow$\\
$\T{\delta}{x}$ & $x(1)1 ::\; \rightarrow$
\end{tabular} 
\\

Remark that Rule~{\Rin} corresponds to Rule~$\T{\beta}{1}$.  Also, note that
since $x$ robots are located on the border node of an $1$.block, 
the local configuration for both {$\T{\gamma}{x}$} and {$\T{\delta}{x}$} is $0x11$.
Similarly, define the generic rule {$\Tsgl{y}$} ($y\geq 2$) as follow: $0(y)0 :: \leftarrow \vee \rightarrow$.

%
\begin{lemma}\label{tour_incompatible-1}
Let $\PR$ be a semi-synchronous protocol for $\phi = 1$ and $5 \leq k < n$. If $\PR$ solves the exploration 
problem, then for every $x\geq 2$, Rule~{$\Tsgl{x}$} is mutually exclusive with both {$\T{\gamma}{x-1}$} and
{$\T{\delta}{x}$}, with respect to $\PR$. 
\end{lemma}
\begin{proof}
Assume that $\PR$ includes {$\Tsgl{x}$} for some $x \geq 2$.  There are two cases to consider:
Assume by contradiction that $\PR$ includes either Rule~{$\T{\gamma}{x-1}$} or Rule~{$\T{\delta}{x}$}.  Then, 
starting from a configuration $0(x-1)11$ or $0(x)11$, {$\T{\gamma}{x-1}$} or {$\T{\delta}{x}$} is enabled, 
respectively.  In both cases, by executing {$\T{\gamma}{x-1}$} or {$\T{\delta}{x}$}, the configuration locally 
becomes $0x01$ in which Rule~{$\Tsgl{x}$} is enabled for each robot belonging to the $x$.tower. 
By executing {$\Tsgl{x}$} on one robot only, the adversary can lead the system in 
a configuration that is undistinguishable from the previous one.  
A contradiction by Theorem~\ref{th:undist}.
\end{proof}


\begin{lemma}\label{tour_incompatible-2}
Let $\PR$ be a semi-synchronous protocol for $\phi = 1$ and $5 \leq k < n$. If $\PR$ solves the exploration 
problem, then for every odd $x\geq 3$, the rules {$\Tsgl{x}$} and {$\T{\beta}{\lfloor \frac{x}{2} \rfloor}$} are mutually 
exclusive with respect to $\PR$. 
\end{lemma}
\begin{proof}
Assume by contradiction that $\PR$ includes both {$\Tsgl{x}$} and {$\T{\beta}{\lfloor \frac{x}{2} \rfloor}$} for some odd 
$x \geq 3$.
Starting from a configuration in which {$\Tsgl{x}$} is enabled, by executing {$\Tsgl{x}$}, the adversary 
can execute {$\Tsgl{x}$} on $x-1$ robots, moving the half of them on the right side and 
the half of them on the left side.  By doing that, the system reaches either a towerless $1$.block or 
an $S^{\lfloor \frac{x}{2} \rfloor}$-sequence, depending on whether $x=3$ or $x\geq 5$, respectively. 
In both cases, {$\T{\beta}{\lfloor \frac{x}{2} \rfloor}$} is enabled on each robot belonging to an extremity 
of the $1$.block (if $x=3$, then {$\T{\beta}{\lfloor \frac{x}{2} \rfloor}$} corresponds to Rule~{\Rin}). 
By executing {$\T{\beta}{\lfloor \frac{x}{2} \rfloor}$} on them, 
the adversary can bring the system in a configuration that is undistinguishable from the previous one.  
A contradiction (Theorem~\ref{th:undist}).
\end{proof}

\begin{lemma}\label{c-general}
Let $\PR$ be a semi-synchronous protocol for $\phi = 1$ and $5 \leq k < n$. If $\PR$ solves the exploration 
problem, then for every $x\geq 1$, $\PR$ includes {$\T{\beta}{x}$} only. 
\end{lemma}
\begin{proof}
The lemma is proven for $x=1$ by Corollary~\ref{cor:init}.  Assume by contradiction that there exists $x\geq 2$ 
such that $\PR$ does not include Rule {$\T{\beta}{x}$} only.  From Lemma~\ref{lem:sur}, there 
exists some executions of $\PR$ leading to a configuration containing an $S^2$-sequence.
Let us consider the smallest $x$ satisfying the above condition.  So, starting from an $S^2$-sequence, 
 for every $y \in [1,x-1]$, by activating Rule {$\T{\beta}{y}$} on every robot at the extremity of the $S^y$-sequence, 
in $x-1$ steps, the adversary can lead the system into a configuration $\gamma_{x-1}$ containing an $S^x$-sequence.  
Since $\PR$ does not include only Rule {$\T{\beta}{x}$}, there are two cases:
\begin{enumerate}
\item $\PR$ includes Rule {$\T{\alpha}{x}$}.  In that case, by activating $x-1$ robots on each $x$.tower 
of the $S^x$-sequence, the system reaches an $S^{x-1}$-sequence.  This configuration is undistinguishable with 
the configuration $\gamma_{x-2}$, the previous configuration of $\gamma_{x-1}$ in the execution. 

\item $\PR$ does not include Rule {$\T{\alpha}{x}$}.  Then, the neighbors of the towers are the only robots that
can be activated with either {$\T{\gamma}{x}$} or {$\T{\delta}{x}$}. 
In the first case ({$\T{\gamma}{x}$} is enabled), the system reaches a configuration similar 
to the one shown in Figure~\ref{YY}, {\em i.e.}, a $1$.block surrounded by two isolated $x+1$.towers. The size of the $1$.block is equal 
to $k-2(x+1)$. By Lemma~\ref{tour_incompatible-1}, $\PR$ does not include
  {$\Tsgl{x+1}$}. So, none of the robots belonging to a tower is enabled.  Furthermore, 
the local view of both extremities of the $1$-block being similar to the initial configuration $\gamma_0$, 
  from Corollary~\ref{cor:init}, both can be activated with Rule~{\Rin} by the adversary.  Then, the system reaches 
  a configuration similar to the one shown in Figure~\ref{x+1-bis}, \ie an $S^2$-sequence, surrounded with $2$ towers 
  and separated by one single empty node. In the second case ({$\T{\delta}{x}$} is enabled), the system reaches a configuration that is also similar to the one shown in Figure~\ref{x+1-bis}. Following the same reasoning as above, 
  the system reaches a configuration similar to the one shown in Figures~(\ref{kequal1}, \ref{kbig3}) or in Figures~(\ref{kequal0}, \ref{kbigeq2})  
depending on whether $k$ is odd or even.

\begin{figure}[H]
 \begin{minipage}[b]{.46\linewidth}
 \centering\epsfig{figure=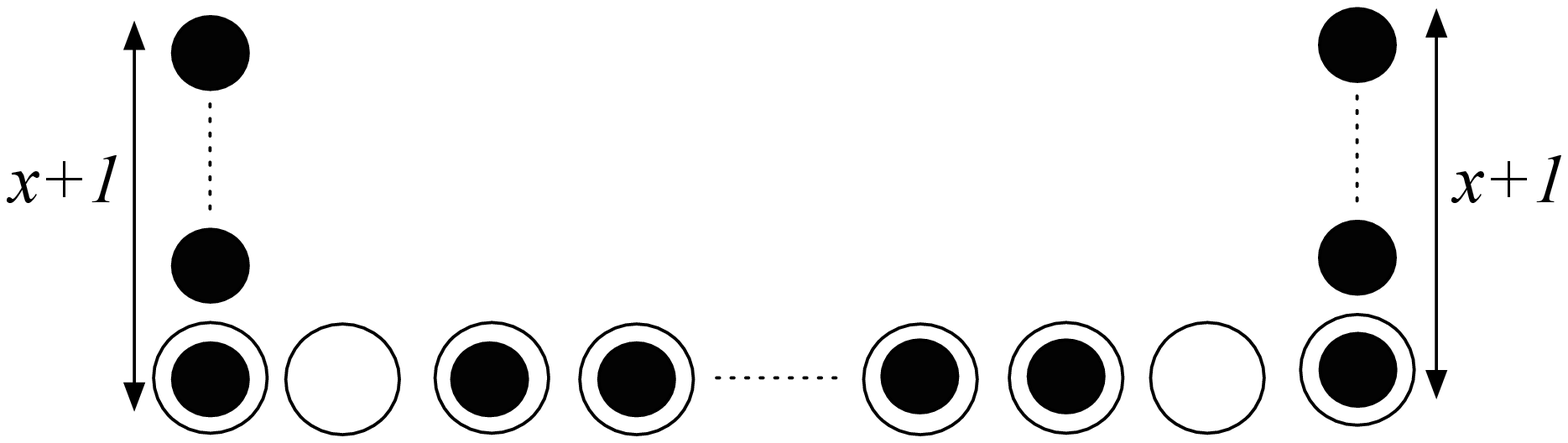,width=5cm}
 \caption{{$\T{\gamma}{x}$} is executed. \label{YY}}
\end{minipage} \hfill
 \begin{minipage}[b]{.46\linewidth}
  \centering\epsfig{figure=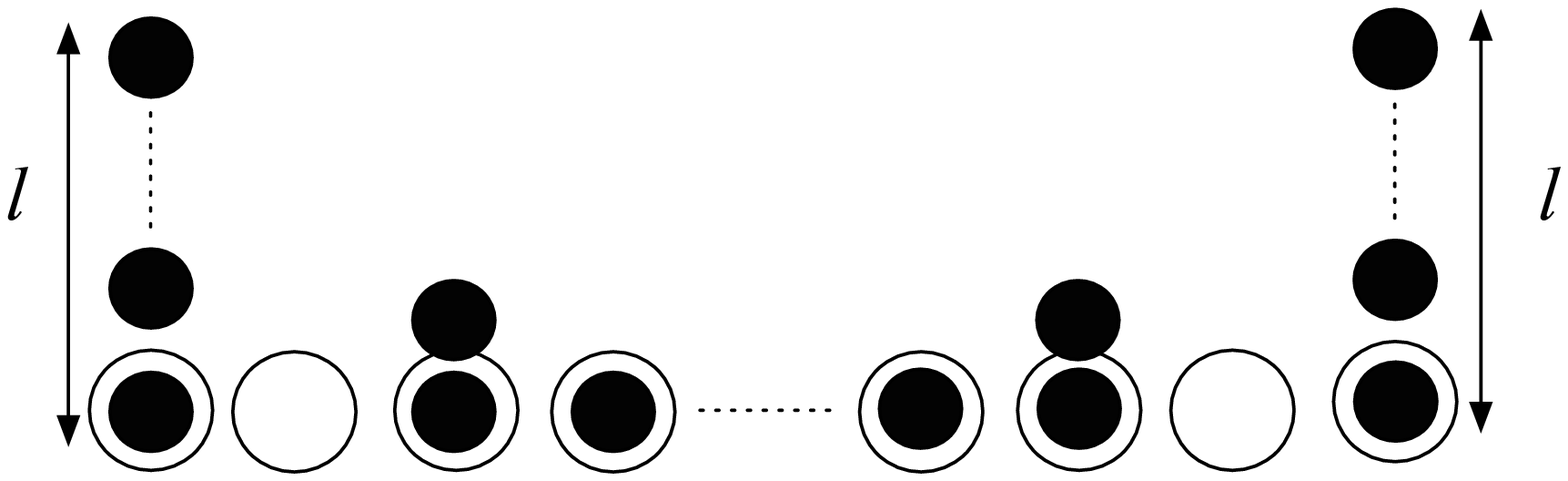,width=5cm}
 \caption{{$\T{\delta}{x}$} is executed \label{x+1-bis}}
 \end{minipage}\hfill
 \begin{minipage}[b]{.46\linewidth}
  \centering\epsfig{figure=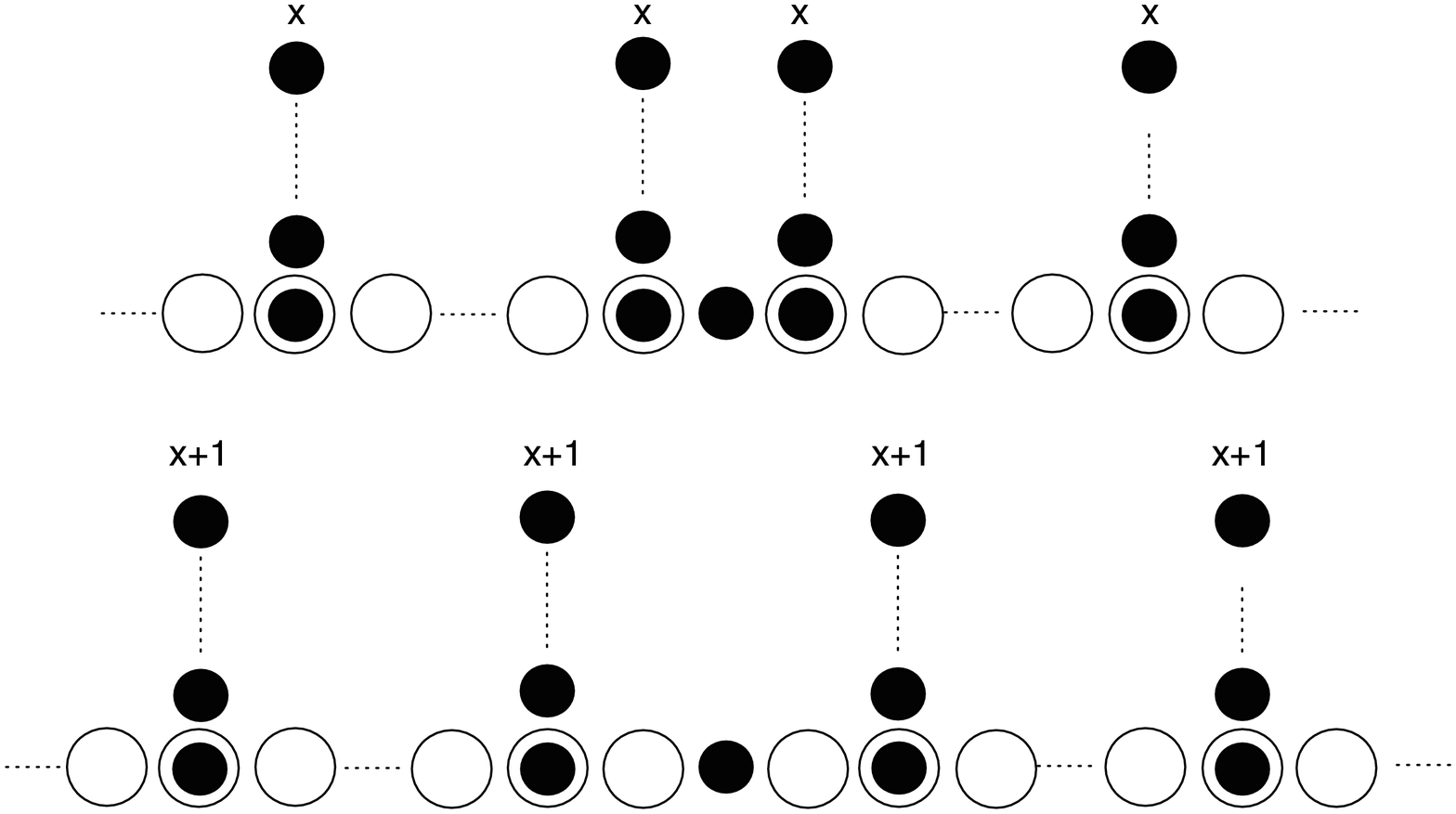,width=6cm}
 \caption{The case where $\kappa=1$  \label{kequal1}}
 \end{minipage}\hfill
 \begin{minipage}[b]{.46\linewidth}
  \centering\epsfig{figure=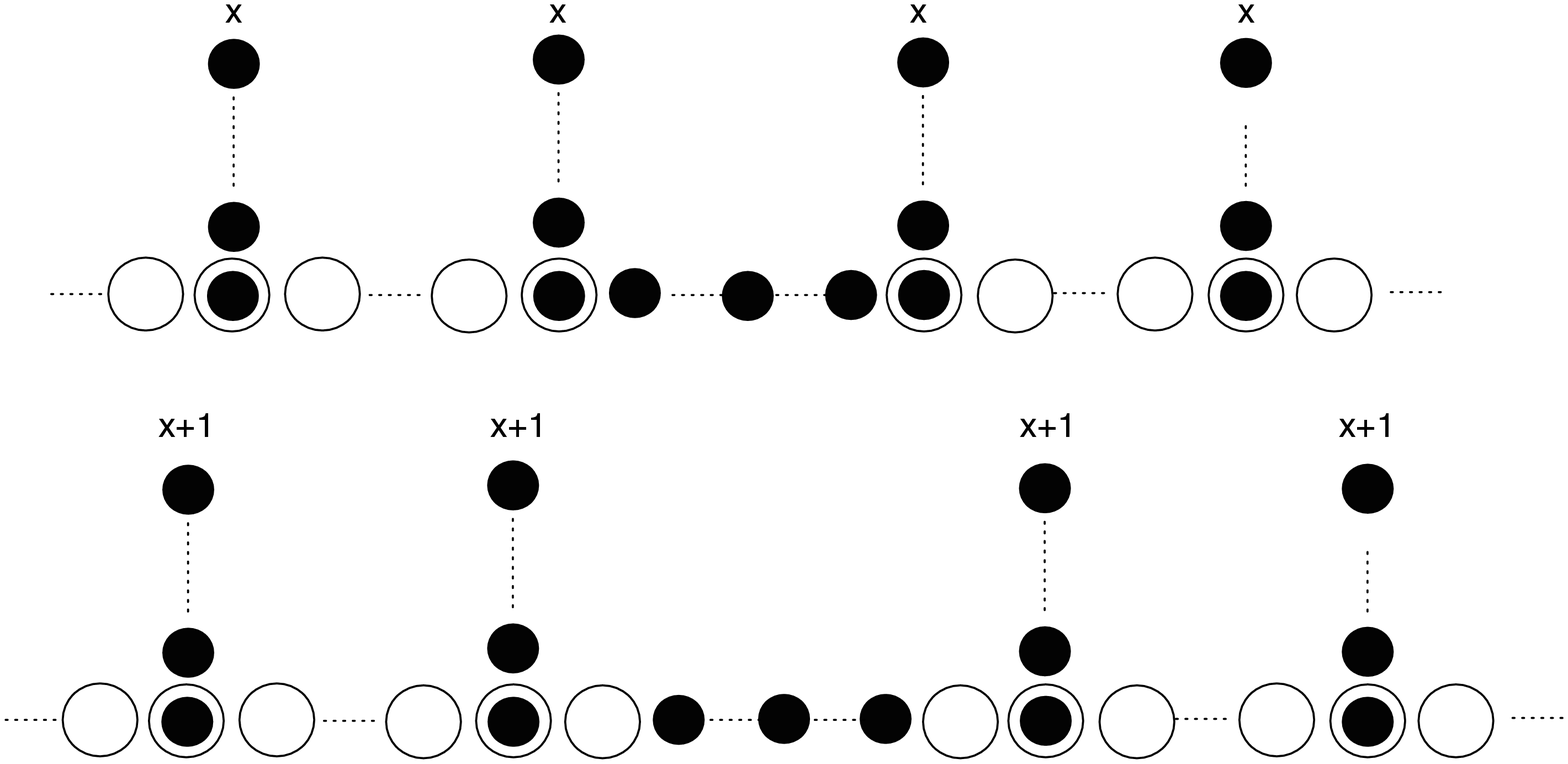,width=6cm}
 \caption{The case where $\kappa$ is odd and $\kappa\geq 3$   \label{kbig3}}
 \end{minipage}\hfill
 \begin{minipage}[b]{.46\linewidth}
  \centering\epsfig{figure=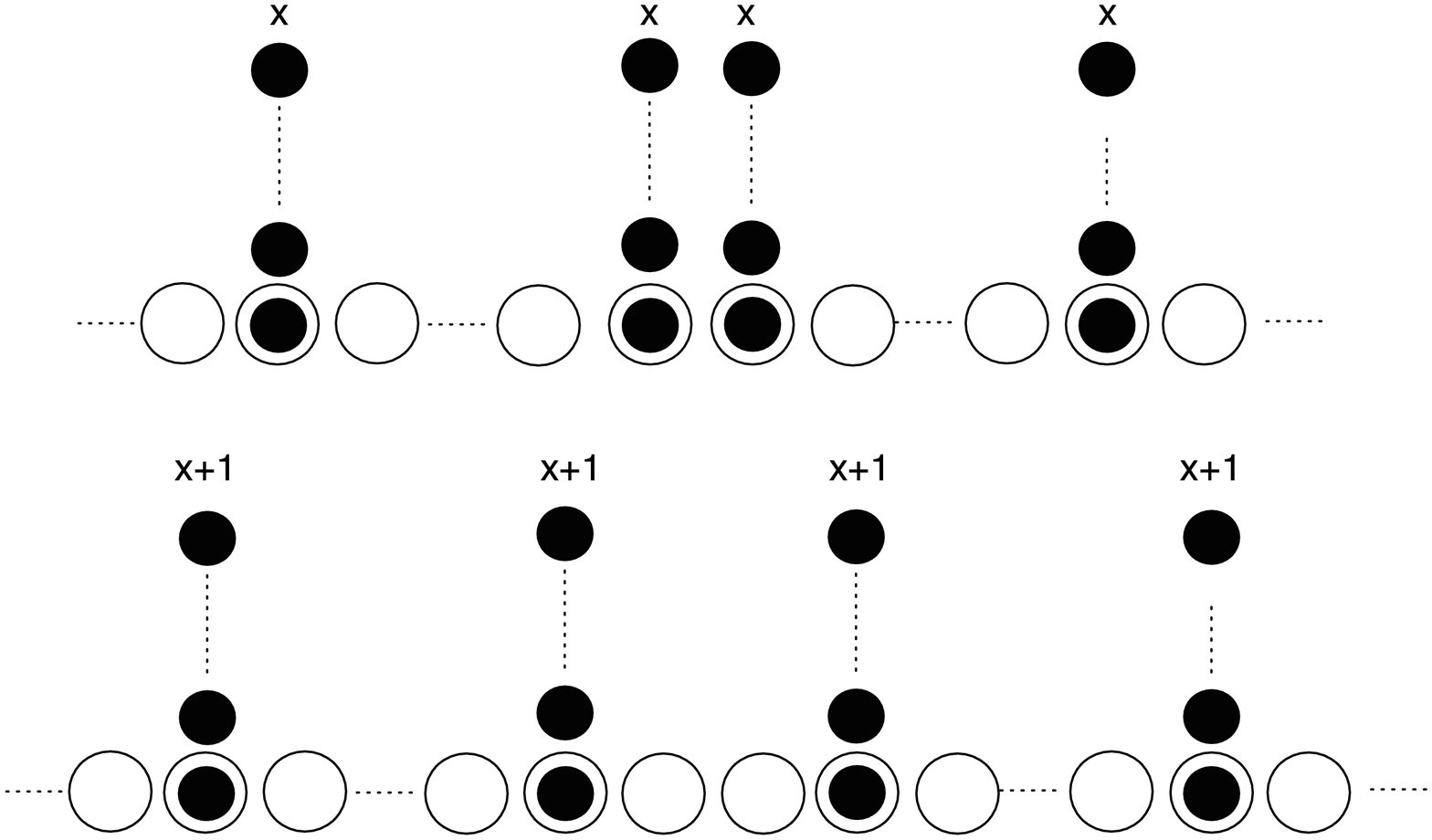,width=6cm}
 \caption{The case where $\kappa=0$  \label{kequal0}}
 \end{minipage}\hfill
 \begin{minipage}[b]{.46\linewidth}
  \centering\epsfig{figure=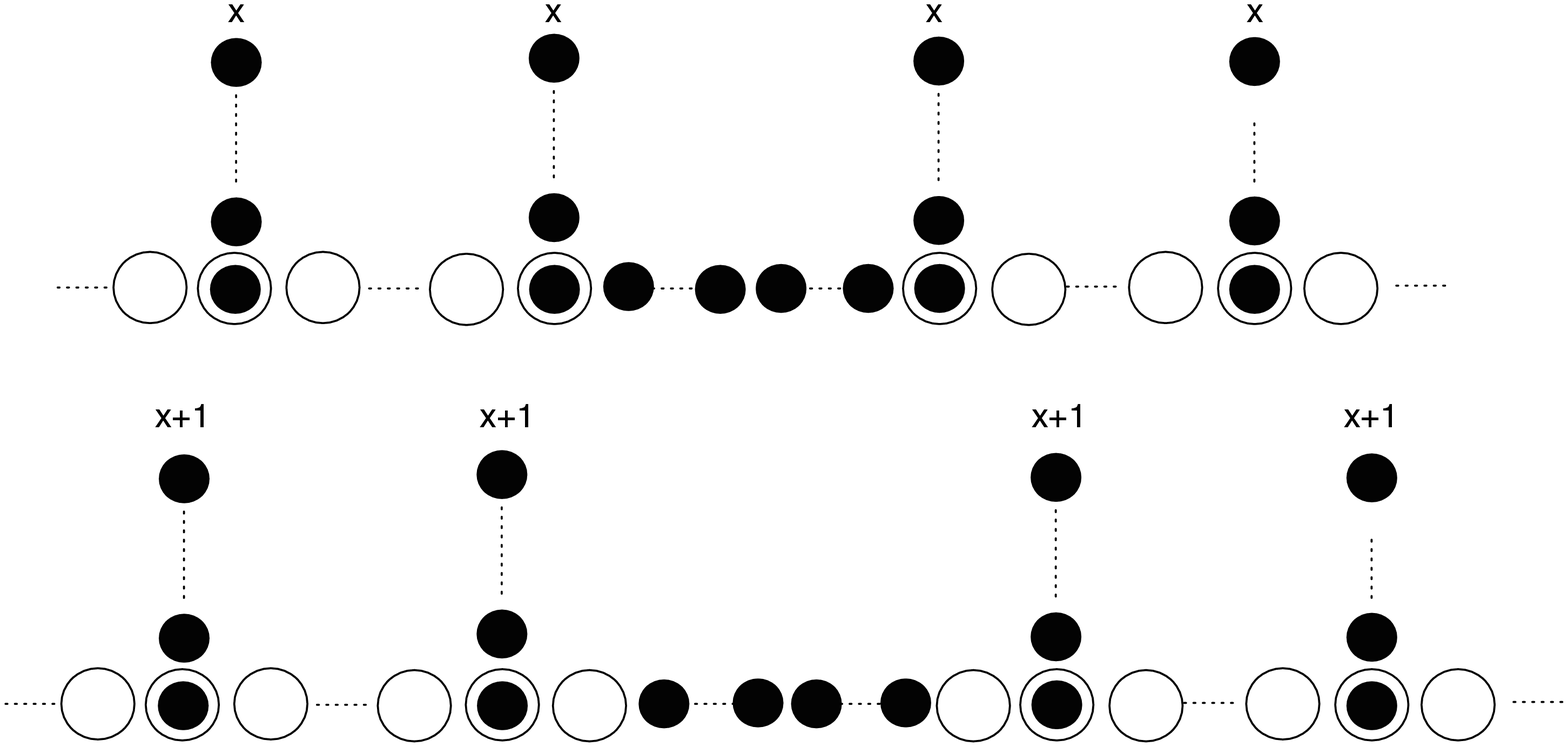,width=6cm}
 \caption{The case where $\kappa$ is even and $\kappa\geq 2$   \label{kbigeq2}}
 \end{minipage}\hfill
\end{figure}

  Let $\tau$, $S$, and $\kappa $ be the number of isolated towers, their size, and the number of robots surrounded 
  by $\tau$ towers, respectively. 
  By construction, $\tau = \lfloor \frac{k}{S} \rfloor$, $\kappa = k \mod (\tau S)$, and $\kappa  \leq 2x$.
  There are two cases to consider:

  \begin{enumerate}
  \item $\kappa $ is odd. Again consider two cases:  
     \begin{enumerate}
     \item $\kappa  \geq 3$.  
     Then, $\lfloor \frac{\kappa }{2} \rfloor < x$ and by assumption, for every $y \in [1,\lfloor \frac{\kappa}{2}\rfloor ]$,
     $\PR$ includes the rules {$\T{\beta}{y}$}.  So, the adversary can lead the system into a configuration containing a single 
     $\kappa +1$.tower surrounded with $\tau$ towers.  Since 
     $\frac{\kappa }{2} < x$ and since by assumption, for every $y \in [2, \frac{\kappa }{2}\rfloor ]$,
     $\PR$ includes Rule~{$\T{\beta}{y}$}, by Lemma~\ref{tour_incompatible-2}, $\PR$ does not include 
     Rule~{$\Tsgl{\kappa}$}.  So, no robot of the $\kappa$.tower is enabled.  So, the system is at a deadlock and
     some nodes are not visited by at least one robot.
     This contradicts that $\PR$ is an exploration protocol. 
     \item $\kappa  = 1$. The only applicable rules on the single robot is either ($i$) Rule~{\Rsgl} or ($ii$) Rule~{\SP} defined as follow: $x(1)x :: \leftarrow \vee \rightarrow$ depending on whether $\PR$ includes {$\T{\gamma}{x}$} or {$\T{\delta}{x}$}, respectively. In the first case ($i$), after executing
     the rule, the single robot is located on the neighboring node of one of the towers, \ie a configuration containing 
     $0(y)10$ where $y$ is equal to $x+1$. 
     The only possible action is then to create an $y+1$.tower. If $\PR$ includes {$\Tsgl{y+1}$}, then by selecting a single 
     robot of the $y+1$.tower, the adversary can lead the system into the previous configuration.  
     A contradiction (Theorem~\ref{th:undist}). In the second case ($ii$), an $x+1$.tower is created. As in Case ($i$), if $\PR$ includes {$\Tsgl{x+1}$}, then by selecting a single 
     robot of the $x+1$.tower, the adversary can lead the system into the previous configuration.  
     A contradiction (Theorem~\ref{th:undist}).
     \end{enumerate}
  \item $\kappa$ is even. There are three cases.  
     \begin{enumerate}
     \item 
        $\kappa  =0$. Since none of the robots of $\tau$ towers is enabled, the system is at a deadlock.  
        This contradicts that $\PR$ is an exploration protocol.
     \item $\kappa  =2$. in this case there is a single $1$.block of size $2$ having either two neighboring $x$.towers at distance $1$ or $x+1$.towers at distance $2$ depending on whether $\PR$ includes {$\T{\gamma}{x}$} or {$\T{\delta}{x}$}, respectively. In both cases, the only applicable rule is $\T{\delta}{x}$. Assume that the scheduler activates both robots at the same time. Then, the system remains into an undistinguishable configuration. A contradiction.
     \item 
        $\kappa > 2$.
     Then, $\frac{\kappa }{2} < x$ and by assumption, for every $y \in [2, \frac{\kappa }{2}\rfloor ]$,
     $\PR$ includes the rules {$\T{\beta}{y}$}, the adversary leads the system into a configuration containing two 
     neighboring $\kappa'$.towers, 
     $\kappa'=\frac{\kappa }{2}$---refer to Figure~\ref{kappabig2}.  
     Assume that $\PR$ includes Rule~{$\Tdbl{\kappa '}$} defined as follows: 
     $0(\kappa')\kappa' :: \leftarrow$.  Then by choosing to execute Rule~{$\Tdbl{\kappa'}$} on $\kappa'-1$ robots
     on each tower (since Rule~{\Rout} cannot be a part of $\PR$, $\kappa > 2$), the adversary brings the system 
     into an $S^{\kappa'-1}$-sequence that is undistinguisable from the
     previous configuration.  If Rule~{$\Tdbl{\kappa'}$} is defined as follows: 
     $0(\kappa')\kappa' :: \rightarrow$. Then by choosing to execute Rule~{$\Tdbl{\kappa'}$} 
     on the same number of 
     robots on each $\kappa'$.tower, the system remains into an undistinguishable configuration.  In both cases,
     this contradicts Theorem~\ref{th:undist}.
     \end{enumerate}
  \end{enumerate}
\end{enumerate}
\end{proof}

\begin{figure}
 \centering\epsfig{figure=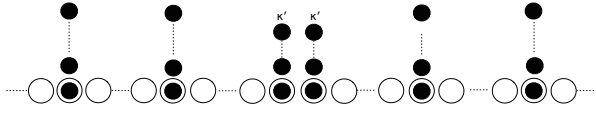,width=8cm}
 \caption{The reached configuration when $\kappa > 2$. }\label{kappabig2}
\end{figure}

Let us now suppose that $\gamma_0$ is any arbitrary starting configuration. The following lemma holds:

\begin{lemma}
Let $\PR$ be a semi-synchronous protocol for $\phi = 1$ and $2\leq k <n$. If $\PR$ solves the exploration problem, 
then Rule {\Rout} and Rule {\Rsgl} are mutually exclusive with respect to $\PR$.
\end{lemma}

\begin{proof}
Assume by contradiction that $\PR$ includes both Rule {\Rout} and Rule {\Rsgl}. Note that Rule {\Rout} (respectively,
Rule {\Rsgl}) is enabled only on robots that are located at the border of an $1$.block (resp., on isolated robots).
Suppose that the initial configuration $\gamma_0$ contains two isolated robots, $r_1$ and $r_2$, that is at distance 
$2$ from a $1$.block, one at each side. 
By executing {\Rsgl} on both $r_1$ and $r_2$, the adversary can bring the system in a configuration $\gamma_1$ where both 
$r_1$ and $r_2$ are located at the border of the $1$.block. Then, Rule {\Rout} becomes enabled on both $r_1$ and $r_2$. 
By executing {\Rout} on both $r_1$ and $r_2$, the adversary brings the system in a configuration $\gamma_2$ that 
is undistinguishable from $\gamma_0$.  A contradiction. 
\end{proof}

\begin{lemma}
Let $\PR$ be a semi-synchronous protocol for $\phi = 1$ and $2\leq k <n$. If $\PR$ solves the exploration problem, then $\PR$ does not include Rule {\Rout}.
\end{lemma}

\begin{proof}
By contradiction, assume that $\PR$ includes Rule {\Rout}. Note that Rule {\Rout} is only enabled on robots that are at the border of a $1$.block. Suppose that $\gamma_0$ contains two $1$.blocks at distance $2$ from each other. Assume that the scheduler activates only one extremity of one of the two $1$.blocks. Let us refer to the $1$.block from which one extremity is activated by $B1$ and let refer to the other $1$.block by $B2$. Let $u_i u_{i+1} \dots u_{i+m}$ be the 1.block $B1$ in $\gamma_0$ and let $u_{i-2} u_{i-3} \dots u_{i-m'}$ be the $1$.block $B2$ ($m$ and $m'$ are the size of respectively $B1$ and $B2$). Denote $r_j$ the robot located on node $u_{i+j}$ ($0\leq i <k$). 
Without lost of generality, assume that the adversary activates $r_0$ only in $\gamma_0$. 
Once $r_0$ moves, it joins $B2$ in $\gamma_1$. Assume that in $\gamma_1$ the scheduler activates $r_1$ that executes
{\Rout} again.  More generally, assume that $\forall$ $i \in [1, \dots, m-2]$, in $\gamma_{1+i}$, the adversary activates
$r_{i+1}$. Once the configuration $\gamma_{k-1}$ is reached, the scheduler activates the same robots in the reverse order
(note that these robots are part of $B2$). At each time, one robot is activated, it joins back $B1$. When $m-2$ robots
are activated, the configuration reached is undistinguishable from $\gamma_1$. A contradiction.
\end{proof}

Observe that both Lemmas \ref{3-4} and \ref{3} are valid when $\gamma_0$ is any towerless configuration. In the case
where there are isolated robots that are at distance $2$ from a $1$.block, the scheduler activates them, such that
they join a $1$.block by executing {\Rsgl}. The only rule that can be enabled is {\Rin}. Each $1$.block, behaves
independently from the other $1$.blocks. Lemmas \ref{lem:kt5} to \ref{c-general} are also valid. Therefore:

\begin{theorem}\label{ATOM1}
No deterministic exploration protocol exists in the semi-synchronous model (ATOM) for $\phi = 1$, $n>1$, and $1 \leq k < n$. 
\end{theorem} 
\begin{proof}
The theorem is a direct consequence of Lemma~\ref{lem:kt5} and Lemma~\ref{c-general} and its proof. 
\end{proof}

Since the set of executions that are possible in CORDA is a strict superset of those that are
possible in ATOM, Theorem~\ref{ATOM1} also holds in the CORDA~\cite{P02}.

\begin{corollary}
No deterministic exploration protocol exists in the asynchronous model (CORDA) for $\phi = 1$, $n>1$, and $1 \leq k < n$. 
\end{corollary} 

 
\paragraph{Synchronous Model.}
\begin{theorem}
Let $\PR$ be a synchronous exploration protocol for $\phi = 1$ and $2\leq k < n$. If $n > 7$, then, 
$k$ must be greater than or equal to $5$.
\end{theorem}

\begin{proof}

The proof is by contradiction. Assume that there exists a  synchronous exploration protocol $\PR$ for $\phi = 1$, 
$2\leq k < 5$, and $n > 7$. 


Let us start from a possible starting configuration, the idea is to derive all the possible executions and then show that non of them can ensure the completion of the exploration task. The cases bellow are possible:

\begin{enumerate}
\item $k=1$.  \label{C1IMP} It is clear that no exploration is possible by a single robot on the ring.
\item $k=2$ \label{C2IMP}. The robots in this case, either move towards each other exchanging their position or, they move in the opposite direction and become isolated robots. In the latter case, once they are activated again, the scheduler brings them back to their previous position. The configuration  reached is in both cases is {\em undistinguishable} from the starting one. A contradiction.
\item $k=3$. $(i)$ In the case where only Rule {\Rout} is enabled. Once the rule is executed, all robots in the configuration become isolated robots. Thus no exploration is possible (all the configurations that are created after are undistinguishable). $(ii)$ If Rule {\Rin} is enabled alone, then the configuration reached contains a single $3$.tower. All robots part of the $3$.tower behave as a single robot. According to Case~(\ref{C2IMP}), no exploration is possible in this case. $(iii)$ Rule {\Rswp}, is the only one enabled. In this case the configuration reached contains an isolated robot and an isolated $2$.tower (let refer to this configuration by $T$). Note that only robots part of the tower are enabled. Suppose that the scheduler activates them such that they move in the opposite direction of the isolated robots, and then, it activates them again but this time, they move in the opposite direction (towards the isolated robot). The configuration reached is undistinguishable from $T$ configuration. Thus no exploration is possible. $(iv)$ If Rules {\Rswp} and {\Rin} are enabled at the same time, each node of the ring is visited by robots by executing the following two rules: $0(2)1:: \rightarrow$ and $2(1)0:: \rightarrow$, however, robots are not able to detect the end of the exploration. Thus, no exploration is possible in this case. Finally, if $(v)$ Rules {\Rswp} and {\Rout} are enabled at the same time. Once such rules are executed, the configuration reached contains one $1$.block of size $2$ and one isolated robots. By executing Rule {\Rout} (possibly  {\Rsgl}), a configuration with just isolated robots is eventually reached. Thus, no exploration is possible in this case too.
\item $k=4$. Suppose that the initial configuration contains a single $1$.block of size $4$. The cases bellow are possible:

\begin{enumerate}
\item Rule {\Rswp} is executed. The two robots inside the $1$.block exchange their position and the configuration reached is {\em undistinguishable} from the starting one. A contradiction.

\item Rule {\Rin} is executed. Two neighboring towers are then created. If the towers move towards each other, the configuration reached is {\em undistinguishable} from the starting one. If they move in the opposite direction of each other, then they become both isolated towers \ie the robots in the tower cannot see any neighboring robot. When they are activated again, the scheduler brings them back to their previous position. Thus, in this case too, the configuration reached is {\em undistinguishable} from the starting one. A contradiction. 

\item \label{CCX} Rule {\Rout} is executed. The robots that were at the extremity of the $1$.block become isolated robots. if Rule {\Rsgl} is enabled then two isolated towers are created (Note that there is at least one node that has not been explored yet). Robots in the same tower behaves as a single robot. Thus, the configuration is similar to the one containing only two robots. According to Case ~(\ref{C2IMP}), no exploration is possible in this case. Hence, Rule {\Rsgl} cannot be enabled. Rule {\Rout} is the only one that is enabled on robots part of the $1$.block of size $2$. Once they move, two $1$.blocks of size $2$ are created. All the robots have now the same view. Note that Rule {\Rout} remains the only one enabled on all the robots. Once it is executed, a new $1$.block of size $2$ is created and the configuration contains either $(i)$ another $1$.block of size $2$ ($n=8$) or $(ii)$ a tower $(n=7)$ or $(iii)$ two isolated robots $(n>8)$. In the first case $(i)$, {\Rout} remains enabled on all the robots. Once it is executed, the configuration reached is undistinguishable from one of the previous ones. Thus, no exploration is possible. In the second case $(ii)$, {\Rout} remains enabled on robots part of the $1$.block. Observe that the robots in the tower cannot be activated otherwise, an undistinguishable configuration is reached (the one that contains two $1$.blocks). Thus only {\Rout} is enabled. Once it is executed, the configuration contains two isolated robots. Such isolated robots can never detect the end of the exploration task since no rule can be executed anymore. Thus no exploration is possible in this case too. In the third case $(iii)$ on robots part of the $1$.block, Rule {\Rout} is enabled, once they move all the robots become isolated robots. Hence, the scheduler can choose always the same direction for all the robots. The configuration reached at each time is then {\em undistinguishable} from the previous one. A contradiction.
 
\item All the robots are allowed to move. Note that in this case either two towers are created and it is clear that no exploration is possible (since each tower behaves as a single robot and thus, the system behaves as there is two robots on the ring). or we retrieve case \ref{CCX}. Thus, no exploration is possible in this case too.  
\end{enumerate}
\end{enumerate}
From the cases above, we can deduce that no exploration is possible for $n>7$ using four robots even in the fully-synchronous model.
\end{proof}


\subsection{Synchronous Algorithms}
\label{sub:synch}

In the following, we present optimal deterministic algorithms that solve the exploration problem in the fully-synchronous model for any ring $3 \leq n\leq 6$ and $n\geq 7$, using respectively $(n-1)$ robots (except for the case $n=6$ that needs only $4$ robots) and five robots. Note that the starting configuration contains a single $1$.block (refer to Figure \ref{Fully-Synch}).

\paragraph{Fully-Synchronous Exploration for $k=5$ and $n\geq 7$.} The idea of Algorithm~\ref{algo:A1} is as follow: The robots that are
at the border of the 1.block are the only ones that are allowed to move in the initial configuration, $\gamma_0$. 
Their destination is their adjacent occupied node (Rule~$1A1$).  Since the system is synchronous, the next
configuration, $\gamma_1$, contains a single robot surrounded by two $2$.towers.
In the next step, the towers move in the opposite direction of
the single robot (Rule~$1A3$) and the single robot moves towards one of the two towers (Rule~$1A2$). 
Note that the resulting configuration $\gamma_2$ is $21020^{n-4}$, providing an orientation of the ring. 
From there, the single $2$.tower are the landmark allowing to detect termination and the three other robots 
explore the ring by perform Rules~$1A3$ and~$2A4$ synchronously. 
After $n-4$ steps, $n-4$ nodes are visited by the $3$ robots and the system reaches $\gamma_{n-2}$ that 
is equal to $2210^{n-3}$.  Finally, by performing Rule~$2A4$, the single robot create a $3$.tower, marking
the end of the exploration in $\gamma_{n-1}$ in which each robot is awake of the termination.  

\begin{algorithm}[htp]
\caption{Fully-Synchronous Exploration for $n\geq 7$ (visibility $1$)}
\label{algo:A1}
\begin{small}

\begin{tabular}{rlcll}

1A1: & $0(1)1$  & $::$ & $\rightarrow$ & \comment{Move towards my occupied neighboring node }
\\
1A2: & $2(1)2$  & $::$ & $\rightarrow$ $\vee$ $\leftarrow$ & \comment{Move towards one of my neighboring node }
\\
1A3: & $0(2)1$ & $::$ & $\leftarrow$ & \comment{Move towards my neighboring empty node}
\\
2A4: & $2(1)0$ & $::$ & $\leftarrow$ & \comment{Move to the tower}

\end{tabular}
\end{small}

\end{algorithm}

\paragraph{Fully-Synchronous Exploration for $3 \leq n\leq 6$.} The formal description of the algorithm is given in Algorithm ~\ref{algo:A'1}. The robots in this case detect the end of the exploration task if they are either part of a $2$.tower or neighbors of a $2$.tower. The idea of the algorithm is the following: For $3 \leq n\leq 5$, $k=(n-1)$ robots are necessary to perform the exploration task. The robots that are at the extremities of the $1$.block are the ones allowed to move, since $k=n-1$, once they move, a $2$.tower is created. If the reached configuration contains an isolated robot (the case where $n=4$) then, this robot is the only one allowed to move. Its destination is one of its adjacent empty nodes (refer to Rule $1A'2$). Once it moves it becomes neighbor of the $2$.tower. In the case where ($n=6$), $4$ robots are necessary to solve the exploration problem. In the same manner, the robots that are at the border of the $1$.block are the ones allowed to move, their destination is their adjacent empty node (refer to Rule $1A'1$). Once they move, the configuration reached contains two $1$.blocks of size $2$. Observe that on all the robots Rule $1A'1$ is enabled. Once it is executed, two $2$.towers are created. The robots can detect the end of the exploration since they are all part of a $2$.tower.
 

\begin{algorithm}[htp]
\caption{Fully-Synchronous Exploration for $3 \leq n\leq 6$ (visibility $1$)}
\label{algo:A'1}
\begin{small}
\begin{tabular}{rlcll}

1A'1: & $0(1)1$  & $::$ & $\leftarrow$ & \comment{Move towards my neighboring empty node}\\
1A'2: & $0(1)0$  & $::$ & $\leftarrow$ $\vee$ $\rightarrow$ & \comment{Move towards one of my neighboring nodes}

\end{tabular}
\end{small}

\end{algorithm}


%
%
%

\begin{figure}[H]
 \begin{minipage}[b]{.46\linewidth}
 \begin{center}
  \epsfig{figure=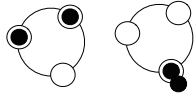,width=2.7cm}
  \end{center}
 \textit{(a) Fully-synchronous exploration $n=3$ and $k=2$} 
 \end{minipage}\hfill
  \begin{minipage}[b]{.46\linewidth}
  \begin{center}
  \epsfig{figure=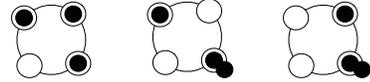,width=5cm}
   \end{center}
 \textit{ (b) Fully-synchronous exploration $n=4$ and $k=3$} 
 \end{minipage}\hfill
 \begin{minipage}[b]{.46\linewidth}
 \begin{center}
 \epsfig{figure=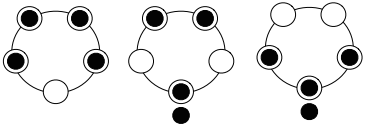,width=5cm}
  \end{center}
 \textit{ (c) Fully-synchronous exploration $n=5$ and $k=4$} 
\end{minipage} \hfill
 \begin{minipage}[b]{.46\linewidth}
 \begin{center}
 \epsfig{figure=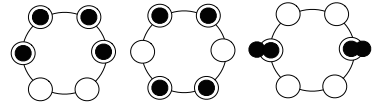,width=5cm}
  \end{center}
\textit{ (d) Fully-synchronous exploration $n=6$ and $k=4$} 
 \end{minipage}\hfill 
 \caption{Fully-Synchronous Deterministic Solutions for $3 \leq n\leq 6$. \label{Fully-Synch}}
\end{figure}


\section{Visibility $\phi=2$}
\label{sec:phi=2}

In this section we present two deterministic algorithms that solve the exploration problem. The first one uses $9$ robots and works for any towerless initial configurations that contains a single $\phi$.group. The second one uses $7$ robots but works only when the starting configuration contains a single 1.block of size $7$. In both algorithms we suppose that $n\geq \phi k+1$.

\subsection{Asynchronous Exploration using $9$ robots}

Let us first define two special configurations: 

\begin{definition}
A configuration is called {\em Middle} (refer to Figure \ref{MIDD}) at instant $t$ if there exists a sequence of consecutive nodes $u_{i}, u_{i+1}, \dots , u_{i+5},u_{i+6}$ such that:
\begin{itemize}
\item $M_j=2$ for $j \in \{i, i+1, i+5, i+6\}$
\item $M_j= 1$ for $j=i+2$
\item $M_j= 0$ for $j \in \{i+3, i+4\}$
\end{itemize}
\end{definition}

\begin{definition}
A configuration is said {\em Terminal} (refer to Figure \ref{TERM}) at instant $t$ if there exists a sequence of nodes $u_{i}, u_{i+1}, \dots , u_{i+4}$ such that: 
\begin{itemize}
\item $M_j=2$ for $j \in \{i, i+3\}$
\item $M_j= 4$ for $j=i+2$
\item $M_j= 0$ for $j=i+1$
\item $M_j= 1$ for $j=i+4$
\end{itemize}
\end{definition}

\begin{figure}
 \begin{minipage}[b]{.46\linewidth}
 \centering\epsfig{figure=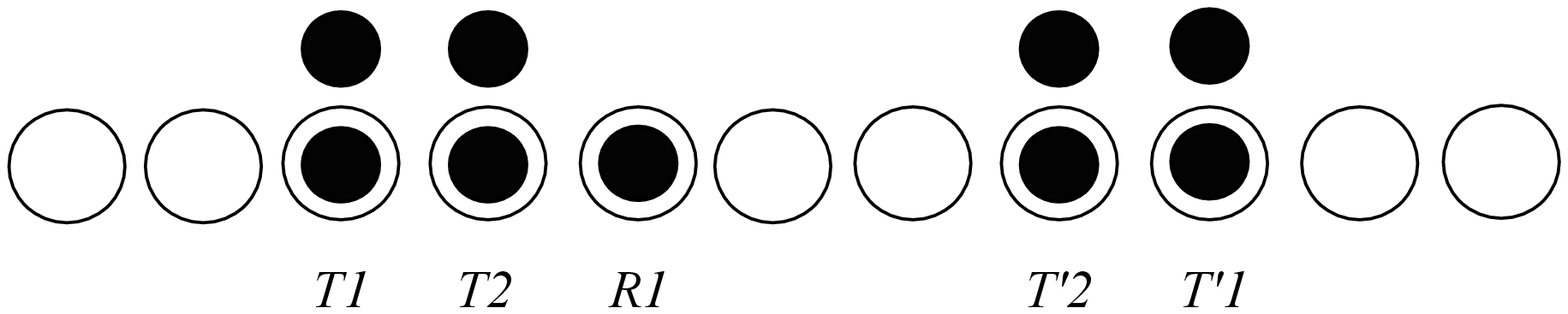,width=5cm}
 \caption{$Middle$ configuration\label{MIDD}}
 \end{minipage} \hfill
\begin{minipage}[b]{.46\linewidth}
  \centering\epsfig{figure=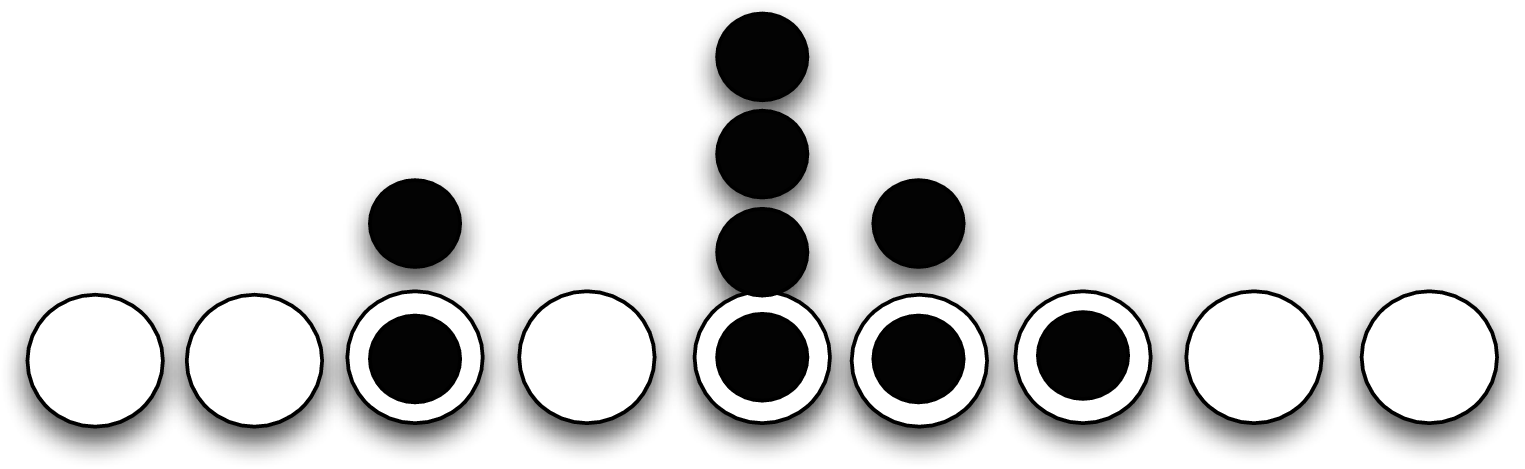,width=5cm}
 \caption{$Terminal$ configuration\label{TERM}}
\end{minipage} \hfill
\end{figure}  

The algorithm comprises two phases as follow:

        \begin{enumerate}
                  \item {\tt Organization Phase}. The aim of this phase is to build the {\tt Middle Configuration}. The initial configuration is any towerless configuration that contains a single $\phi$.group.
                  \item {\tt Exploration Phase}. The starting configuration of this phase is the {\tt Middle Configuration}. A set of robots are elected to perform the exploration while other robots stay still as a land mark. {\tt Terminal} configuration is created at this end of this phase to indicate the end of exploration task.
       \end{enumerate}
       
The formal description of the algorithm is presented in Algorithms \ref{algo:A29}.

\begin{algorithm}[H]
\caption{Asynchronous Exploration using $9$ robots ($\phi=2$)}
\label{algo:A29}
\begin{scriptsize}
\begin{tabular}{rlcll}
\multicolumn{5}{l}{{\tt Organization} Phase}\\
2A'1: & $00(1)01$ & $::$ & $\rightarrow$ & \comment{Move toward the occupied node}
\\
2A'2: & $00(1)1?$ & $::$ & $\rightarrow$ & \comment{Move to my adjacent occupied node}
\\
2A'3: & $00(2)01$ & $::$ & $\rightarrow$ & \comment{Move toward the occupied node}
\\
2A'4: & $02(1)01$ & $::$ & $\rightarrow$ & \comment{Move in the opposite direction  of the tower}
\\
2A'5: & $21(1)01$ & $::$ & $\rightarrow$ & \comment{Move in the opposite direction of the tower}
\\
2A'6: & $21(1)1?$ & $::$ & $\leftarrow$ & \comment{Move toward the tower}
\\
2A'7: & $20(1)02$ & $::$ & $\leftarrow$ $\vee$ $\rightarrow$ & \comment{Move to one of my neighboring nodes}
\medskip\\
\multicolumn{5}{l}{{\tt Exploration} Phase}\\
2A'6: & $21(1)1?$ & $::$ & $\leftarrow$ & \comment{Move toward the tower}
\\
2A'8: & $00(2)21$ & $::$ & $\leftarrow$ & \comment{Move in the opposite direction of my neighboring tower}
\\
2A'9: & $01(1)21$ & $::$ & $\leftarrow$ & \comment{Move in the opposite direction of the tower}
\\
2A'10: & $20(2)10$ & $::$ & $\leftarrow$ & \comment{Move toward the tower at distance $2$}
\\
2A'11: & $20(1)00$ & $::$ & $\leftarrow$ & \comment{Move toward the tower}
\\
2A'12: & $20(2)21$ & $::$ & $\leftarrow$ & \comment{Move to my adjacent free node}
\\
2A'13: & $21(1)21$ & $::$ & $\leftarrow$ & \comment{Move in the opposite direction of my neighboring tower}
\\
2A'14: & $02(2)22$ & $::$ & $\rightarrow$ & \comment{Move toward the tower having another tower as a neighbor}
\\
2A'15: & $02(1)32$ & $::$ & $\rightarrow$ & \comment{Move toward the $3$.tower}
\end{tabular}
\end{scriptsize}
\end{algorithm}

\paragraph{Proof of correctness}

We prove in this section the correctness of our algorithm presented above. In the following, a configuration is called {\tt intermediate} if Rule $2A'7$ is enabled in the configuration.

\begin{lemma}\label{Intermediate}
Starting from any towerless configuration that contains a single $\phi$.group, {\tt intermediate} configuration is eventually reached.
\end{lemma}

\begin{proof}
Two cases are possible as follow:

\begin{itemize}
\item The initial configuration contains a 1.block of size $9$. Rule $2A'2$ is enabled on both robots that
are at the extremity of the 1.block. When the scheduler activates such robots, a $S^2$-sequence is created. Rule $2A'6$ becomes enabled on the robots at distance $2$ from the $2$-towers. When such robots are activated by the scheduler, a configuration of type {\tt intermediate} is created. 
Since we consider in this paper asynchronous robots, one tower can be created only at one extremity of the 1.block,
Rule $2A'6$ becomes enabled on the robots at distance $2$ from such a $2$-tower. However, Rule $2A'2$ keeps being
enabled on the other extremity (let the robot at the other extremity be $r_1$). Since we consider a fair scheduler, $r_1$ will be eventually enabled and a $2$-tower is created. Rule $2A'6$ becomes enabled on the robots at distance $2$ from the new $2$-towers. Once such robot is activated, {\tt intermediate} configuration is created. 

\item Other configurations. Two sub-cases are possible:
             \begin{itemize}
             \item \textbf{Case (a)}: Robots at the extremity of the $\phi.group$ do not have any neighboring robots: Rule $2A'1$ is enabled on such robots. When the scheduler activates them, they become neighbor of an occupied node and we retrieve Case (b).
             \item \textbf{Case (b)}: Robots at the extremity of the $\phi.group$ have one occupied node as a neighbor. Rule $2A'2$ becomes enabled on such robots. When the scheduler activates them, a $2$-tower is created at at least one extremity of the $\phi.group$.
             \end{itemize}
             
When a $2$-tower is created, four sub cases are possible as follow:
             \begin{itemize}
             \item \textbf{Case (c)}: The tower does not have a neighboring occupied node (there is a sequence of nodes with multiplicities equal to $00201$). In this case, Rule $2A'3$ becomes enabled on the robots part of the $2$-tower. If the scheduler activates both robots at the same time then the $2$-tower becomes neighbor to an occupied node. $(i)$ If the scheduler activates only one robot, then the tower is destroyed and Rule $2A'2$ becomes enabled on the robot that did not move. When the rule is executed, the tower is built again and the configuration reached is exactly the same as in $(i)$.  
             
              \item \textbf{Case (d)}: The Tower is neighbor of an occupied node that is part of a 1.block of size $3$ ($002111$). Rule $2A'6$ is enabled in this case on the robot that is in the middle of the 1.block. Once the Rule is executed, another $2$-tower is created.
              
               \item \textbf{Case (e)}: The Tower is neighbor of an occupied node that is part of a 1.block of size $2$ ($002110$). Rule $2A'5$ is enabled in this case on the robot that is in the extremity of the 1.block not having the $2$-tower as a neighbor. When the rule is executed the robot that has moved becomes part of another 1.block and we retrieve Case (f).
             
             \item \textbf{Case (f)}: The Tower is neighbor of an occupied node that has an empty node as a neighbor ($00210$). Rule $2A'4$ becomes enabled on the robot that is at distance $1$ from the tower. Once it moves, it becomes neighbor of another occupied node and we retrieve either Case (c). 
             \end{itemize}
             
 From the cases above we can deduce that a configuration with a sequence of node $2111$ is reached in a finite time (Cases (c), (e) and (f)). In this case, Rule $2A'6$ becomes enabled on the robot that is in the middle of the 1.block of size $3$ having a neighboring $2$-tower (Case (d)). Once the robot moves, a $2$-tower is created ($2201$). Note that the single robot in the sequence is not allowed to move unless it sees a $2$-towers at distance $2$ at each side. Thus, in the case the scheduler activates robots in only one extremity of the $\phi.group$, we are sure that robots in the other extremity will be eventually activated since they are the only one that will be allowed to move (the scheduler is fair). Hence, a symmetric configuration in which there will be one robot on the axes of symmetry and two $2$-tower at each side of it is reached in a finite time. {\tt intermediate} configuration is then reached and the lemma holds.

\end{itemize}

\end{proof}

\begin{lemma}\label{MIDDLE}
Starting from {\tt intermediate} configuration, {\tt Middle} configuration is eventually reached.
\end{lemma}

\begin{proof}
When the configuration is of type {\tt intermediate}, only $2A'7$ is enabled. Once the rule is executed, the robot that is not part of any tower becomes neigbor of a tower of size $2$. {\tt Middle} configuration is reached and the lemma holds.
\end{proof}

\begin{lemma}\label{STM}
Starting from any towerless configuration that contains a single $\phi$.group, a configuration of type {\em Middle} is eventually reached.
\end{lemma}

\begin{proof}
Directly follows from Lemmas \ref{MIDDLE} and \ref{STM}. 
\end{proof}

\begin{lemma}\label{MTT}
Starting from a configuration of type {\tt Middle}, {\tt Terminal} configuration is eventually reached.
\end{lemma}

\begin{proof}
Let $T'1$ and $T'2$ be the two neigboring $2$-towers that cannot see any other robot, and let $T1$ and $T2$ be the other $2$-towers such as $T2$ has an neighboring occupied node other then $T1$ (let refer to the robot that is neighbor to $T1$ and not part of a $2$-tower by $R1$) (refer to Figure \ref{MIDD}). The  sequence of robots containing $T1$, $T2$ and $R1$ is called {\em Explorer-Sequence}.
On such a configuration, Rule $2A'8$ is enabled on $T1$. When the rule is executed, either ($i$) $T1$ becomes at distance $2$ from $T2$ (the scheduler activates both robots part of $T1$ at the same time) or ($ii$) $T1$ is destroyed and in this case Rule $2A'9$ becomes the only rule enabled. When $2A'9$  is executed, $T1$ is restored and becomes at distance $2$ from $T2$ (we retrieve Case ($i$)). On the new configuration Rule $2A'10$ is enabled on $T2$. When the rule is executed, either $(a)$ $T2$ becomes neighbor to $T1$ (the scheduler activates both robots in $T1$) or $(b)$ $T2$ is destroyed and Rule $2A'6$ becomes the only one enabled (the scheduler activates only one robot in $T2$). When the rule is executed, $T2$ is built again and becomes neighbor of $T1$ (we retrieve Case $(a)$). Rule $2A'11$ becomes then enabled on $R1$. When $R1$ is activated, {\em Explorer-Sequence} is built again, however, the distance between $T1$ and $T'1$ has decreased. Since the view of robots in the {\em Explorer-Sequence} is the same as in the {\tt Middle} configuration, they will have the same behavior. Thus, at each time the distance between $T1$ and $T'1$ decreases. Hence $T1$ becomes eventually at distance $2$ from $T'1$ in a finite time. Rule $2A'12$ becomes enabled on $T1$. If the scheduler activates only one robot in $T1$ then $T1$ is destroyed and  Rule $2A'13$ becomes enabled, when the rule is executed, $T1$ becomes neighbor of $T'1$ (it is like the scheduler activates both robots at the same time). Since robots in $T2$ cannot see the $2$-tower $T'1$, they will move towards $T1$, the same holds for $R1$. Rule $2A'14$ is then enabled on $T2$, once it is executed, either a $4$-tower is created and in this case the configuration is of type {\tt Terminal} or $3$-tower is created and in this case Rule $2A'15$ becomes enabled, when it is executed a $4$-tower is created and the configuration becomes $Terminal$ and the lemma holds.
\end{proof}

\begin{lemma}
Starting from any towerless configuration that contains a single $\phi$.group, {\tt Terminal} configuration is eventually reached such that all nodes of the ring have been explored. 
\end{lemma}

\begin{proof}
From both lemma \ref{STM} and \ref{MTT} we can deduce that starting from any initial configuration, {\tt Terminal} is reached in a finite time. In another hand, from Lemma \ref{STM} we are sure that {\tt Middle} configuration is reached in a finite time. Note that the only way to build the second towers ($T2$ and $T2'$) is to have a tower neighbor of a consecutive sequence of three robot. Since the distance between $R1$ and $T'2$ is equal to $3$, we are sure that these nodes have been occupied and thus they have been explored. In another hand $T1$, $T2$ and $R1$ keep moving towards $T1'$ until they will become neighbor. Hence, all the nodes that were between $T1$ and $T1'$ have been all visited. We can deduce then that all the nodes of the ring have been explored and the lemma holds.   
\end{proof}

\subsection{Asynchronous Exploration using $7$ robots}

Recall that the initial configuration in this case contains a $1$.block of size $7$.  The idea is similar to the one used in Algorithm \ref{algo:A29}. Let us first define two spacial configurations:

\begin{definition}
A configuration is called {\tt Inter} if there exists a sequence of seven nodes $u_i, u_{i+1}, \dots $\\
$, u_{i+5}, u_{i+6}$ (refer to Figure \ref{INTER}) such that: 

\begin{itemize}
\item $M_j=2$ for $j \in \{i, i+1, i+6\}$
\item $M_j=0$ for $j \in \{i+3, i+4, i+5\}$
\item $M_j=1$ for $j =i+2$
\end{itemize}
\end{definition}

\begin{definition}
A configuration is called {\tt Final} if there exists a sequence of four nodes $u_i, u_{i+1}, \dots $\\
$, u_{i+3}$  (refer to Figure \ref{FNL}) such that: 

\begin{itemize}
\item $M_j=2$ for $j =i$
\item $M_j=4$ for $j =i+2$
\item $M_j=0$ for $j=i+1$
\item $M_j=1$ for $j=i+3$
\end{itemize}
\end{definition}

\begin{figure}
\begin{minipage}[b]{.46\linewidth}
  \centering\epsfig{figure=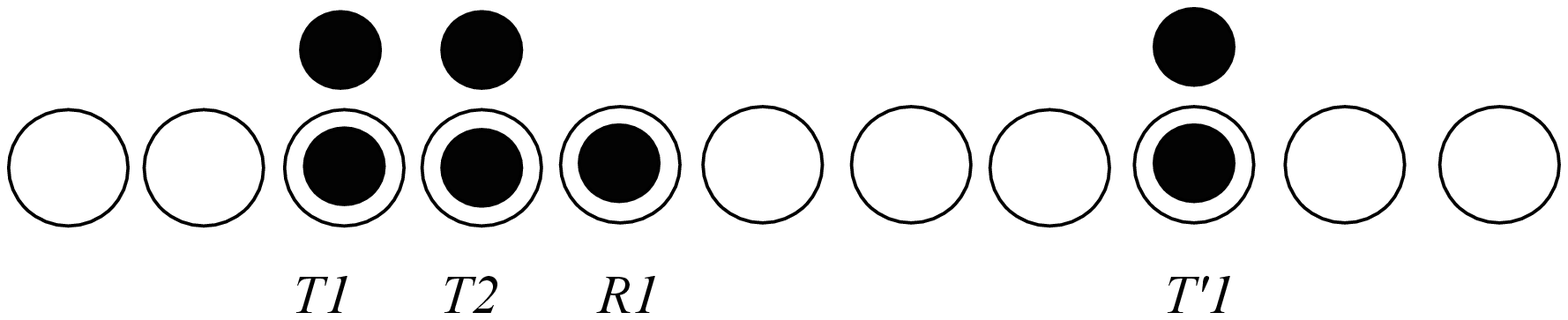,width=4cm}
 \caption{$INTER$ \label{INTER}}
 \end{minipage} \hfill
\begin{minipage}[b]{.46\linewidth}
  \centering\epsfig{figure=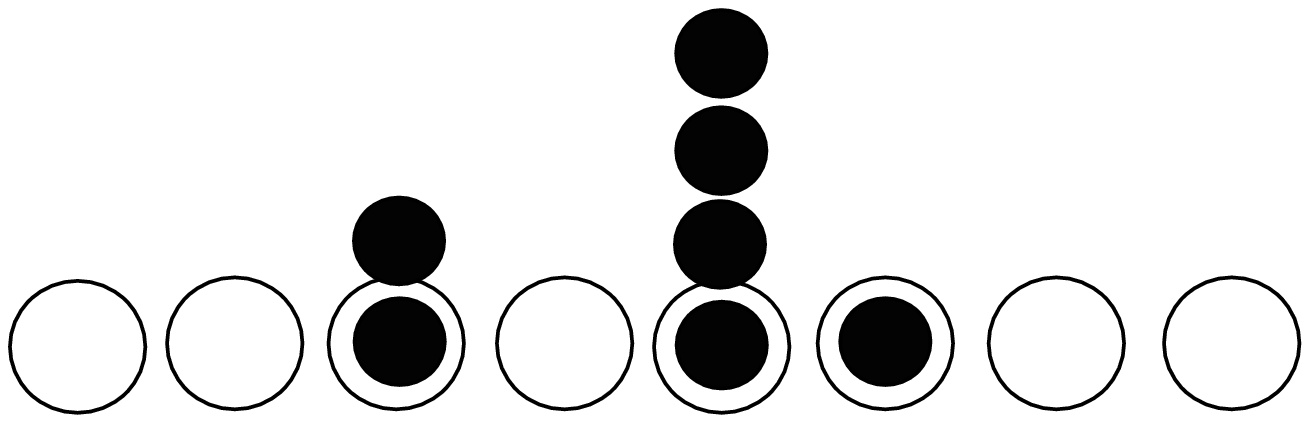,width=4cm}
 \caption{$Final$ configuration\label{FNL}}
\end{minipage}\hfill
\end{figure}

The algorithm comprises two phases as follow:

\begin{enumerate}
\item {\tt Preparation Phase:} The aim of this phase is to build an {\tt Inter} Configuration. The starting configuration contains a single 1.block of size $7$.
\item {\tt Exploration Phase}: In this phase a set of robots are elected to perform the exploration task. At the end of the exploration {\tt Final} Configuration is built to indicate the end of the exploration task.
\end{enumerate}

The formal description of the algorithm is given in Algorithm \ref{algo:A2bis}.

\begin{algorithm}[H]
\caption{Asynchronous Exploration using $7$ robots ($\phi=2$)}
\label{algo:A2bis}
\begin{scriptsize}

\begin{tabular}{rlcll}
\multicolumn{5}{l}{{\tt Preparation} Phase}\\
2A1: & $01(1)11$ & $::$ & $\leftarrow$ & \comment{Move to my adjacent node having a free node as a neighbor}
\\
2A2: & $20(1)1?$ & $::$ & $\leftarrow$ & \comment{Move to my adjacent free node}
\\
2A3: & $10(1)01$ & $::$ & $\leftarrow$ $\vee$ $\rightarrow$ & \comment{Move to one of my neighboring node}  
\\
2A4: & $00(1)20$ & $::$ & $\leftarrow$ & \comment{Move to my adjacent free node}
\\
2A5: & $10(1)02$ & $::$ & $\leftarrow$ & \comment{Move in the opposite direction of the tower}
\\
2A6: & $21(1)10$ & $::$ & $\leftarrow$ & \comment{Move toward the $2$.tower}
\\
2A7: & $20(1)00$ & $::$ & $\leftarrow$ & \comment{Move towards the $2$.tower}
\medskip\\
\multicolumn{5}{l}{{\tt Exploration} Phase}\\
2A6: & $21(1)10$ & $::$ & $\leftarrow$ & \comment{Move toward the $2$.tower}
\\
2A7: & $20(1)00$ & $::$ & $\leftarrow$ & \comment{Move toward the $2$.tower}
\\
2A8: & $00(2)21$ & $::$ & $\leftarrow$ & \comment{Move to my adjacent free node}
\\
2A9: & $01(1)21$ & $::$ & $\leftarrow$ & \comment{Move to my adjacent node having one free node as a neighbor}
\\
2A10: & $20(2)10$ & $::$ & $\leftarrow$ & \comment{Move to my adjacent free node}
\\
2A11: & $20(2)21$ & $::$ & $\leftarrow$ & \comment{Move to my adjacent free node}
\\
2A12: & $21(1)21$ & $::$ & $\leftarrow$ & \comment{Move in the opposite direction of the neighboring $2$.tower}
\\
2A13: & $02(2)21$ & $::$ & $\rightarrow$ & \comment{Move to my adjacent node having one robot as a neighbor}
\\
2A14: & $02(1)32$ & $::$ & $\rightarrow$ & \comment{Move toward the $3$.tower}
\end{tabular}
\end{scriptsize}
\end{algorithm}

\paragraph{Proof of correctness}
We prove in the following the correctness of our algorithm.

\begin{lemma}\label{1BIN}
Starting from a configuration that contains a single 1.block, a configuration of type {\em Inter} is eventually reached.
\end{lemma}

\begin{proof}
When the configuration contains a single 1.block of size $7$, Rule $2A1$ is enabled on the robots that are neighbor of the extremity of the $1$.block. Once such robots move, a $2$-tower is created at each extremity of the 1.block. Rule $2A2$ becomes enabled on the robots that are at distance $2$ from the $2$-towers. Once they move, they become neighbor of a $2$-tower. Note that in the case the scheduler activates only robots at one extremity, we are sure that the configuration that will be reached in this case is exactly the same as when the scheduler activates both robots at each extremity since no robot will be able to move besides the ones at the other extremity. In another hand, since robots that did not move in the other extremity are not aware of the changes in the current configuration (recall that they are only able to see at distance $2$), they will have the same behavior as when they are activated at the same time with the robots at the other extremity.
Rule $2A3$ becomes then enabled. Since the view of the robot on which Rule $2A3$ is enabled, is symmetric the scheduler is the one that chooses the direction to take. Once such robot moves, it becomes neighbor of a single robot and it will be able to see one $2$-tower. In the resulting configuration, Rule $2A4$ becomes enabled on the robot that is neighbor of a $2$-tower having an empty node as a neighbor. Once this robot moves, it becomes able to see the other part that it could not see before moving. This robot continue to move towards the single robot until they become neighbor (refer to Rule $2A5$). Observe that one tower remains unseen by the other robots. The robot that is in the middle of the $1$.block of size $3$ that is neighbor of the tower of a $2$-tower is the only one that can move (refer to Rule $2A6$). When this robot moves, a new $2$.tower is created. The robot that is at distance $2$ from the new $2$.tower (note that there is only one such robot) is the only robot allowed to move, its destination is its adjacent empty node towards the $2$.tower (refer to Rule $2A7$). When it moves it becomes neighbor of the tower of size $2$. {\tt Inter} configuration is then reached and the Lemma holds.
\end{proof}

\begin{lemma}\label{INTFL}
Starting from a configuration of type {\tt INTER}, a configuration of type {\tt Final} is eventually reached.
\end{lemma}

\begin{proof}
When {\tt Inter} configuration is reached, a subset of robots are elected to perform the exploration task. These robots are $T1$, $T2$ and $R1$  (refer to Figure \ref{INTER}). Observe that these robots form exactly the same sequence ({\em Explorer-Sequence}) as in  {\tt Middle} configuration (refer to Algorithm \ref{algo:A29}). It has been shown in Lemma \ref{MTT} that this subset of robots keep moving towards $T'1$ until they become neighbor of $T'1$. Rule $2A13$ becomes then enabled on the robots part of $T1$. If the scheduler activates both robots at the same time, a $4$-tower is created and the configuration reached is {\tt Final}. In the case the scheduler activates only one robot from $T1$ then a $3$-tower is created. There will be one robot not part of a tower that is neighbor to the $3$-tower. This robot is the one allowed to move (refer to Rule $2A14$). Once the rule is executed a $4$-tower is created and the configuration is {\em Final}. Thus, the lemma holds.
\end{proof}

\begin{lemma}
Starting from a configuration that contains a single 1.block of size $7$, a configuration of type $Final$ is eventually reached such that all the nodes of the ring have been explored.
\end{lemma}

\begin{proof}
From Lemma \ref{1BIN} and \ref{INTFL} we can deduce that starting from a configuration that contains a single 1.block, {\tt Final} configuration is reached in a finite time. In another hand, when {\tt Inter } configuration is reached, we are sure that the nodes between $R1$ and $T'1$ have been explored since these nodes have been occupied by robots at the beginning by the 1.block (refer to Figure \ref{INTER}). The sub set of robots that have performed the exploration task ($T1$, $T2$ and $R1$) keep moving towards $T'1$ until it becomes neighbor of $T'1$. Thus we are sure that the nodes that were between $T1$ and $T'1$ in {\tt Inter} configuration have been visited. Thus all the nodes of the ring have been explored and the lemma holds. 
\end{proof}

\section{Visibility $\phi=3$}
\label{sec:phi=3}

In this section we first prove that no exploration is possible with $4$ robots when $n>13$. We then present two deterministic algorithms that solves the exploration. The first one uses $7$ robots and works starting from any towerless configuration that contains a single $\phi$.group. The second one uses $5$ robots and works only when the starting configuration contains a single 1.block of size $5$. In both solutions, we suppose that $n\geq k\phi+1$.

\subsection{Negative Results}

%
%

In the following, let us consider a configuration $\gamma_t$ in which there is a $3$.group of size $4$ (let call this configuration {\em Locked}). Let $u_i, u_{i+1}, \dots, u_{i+9}$ be the nodes part of  the $3$.block in $\gamma_t$ such as $u_i, u_{i+3}, u_{i+6},u_{i+9}$ are occupied. Let $r_1$ and $r_4$ be the robots that are at the extremity of the $3$.block such tha $r_1$ is on $u_i$ and let $r_2$ and $r_3$ be the robots that are inside the $3$.block such that $r_2$ is on $u_{i+3}$. The following three rules are the only ones can be enabled:

\begin{enumerate}
\item \label{3R1} 000(1)001 :: $\leftarrow$
\item \label{3R2} 000(1)001 :: $\rightarrow$
\item \label{3R3} 100(1)001 :: $\leftarrow$ $\vee$ $\rightarrow$
\end{enumerate}

\begin{lemma}\label{L3R1}
Let $\PR$ be a semi-synchronous exploration protocol for $\phi = 3$, $n>13$, and $k<n$.  Then, $\PR$ does not include Rule~\ref{3R1}.
\end{lemma}

\begin{proof}
By contradiction, assume that $\PR$ includes Rule~\ref{3R1}. Note that  Rule~\ref{3R1} is only enabled on $r_1$ and $r_4$. Assume that the adversary activates only $r_1$ in$\gamma_t$. Let us refer to the resulting configuration by $ST$ (standing for "Second Trap" configuration). Once $r_1$ moves, it becomes an isolated robot on $u_{i-1}$ in $ST$. $u_{i+3}$ becomes the border of the $3$.group. 
Assume that in $ST$, the adversary activates $r_2$ (on $u_{i+1}$) that executes Rule~\ref{3R1}.  
Then, Configuration~$ST+1$ includes two $3$.groups disjoint by three nodes.  The former includes $r_1$ and $r_2$ 
(on $u_{i-1}$ and $u_i$, respectively).  The latter forms the sequence $r_3\ldots r_4$, located on $u_{i+6} \ldots
u_{i+9}$. Suppose that the scheduler activates $r_3$ that is on node $u_{i+6}$ that executes Rule ~\ref{3R1}. Once it moves, the resulting 
configuration $ST+2$ is undistinguishable from $ST$.  A contradiction (Theorem~\ref{th:undist}). 
\end{proof}

\begin{lemma}\label{L3R2}
Let $\PR$ be a semi-synchronous exploration protocol for $\phi = 3$, $n>13$, and $k<n$.  Then, $\PR$ does not include Rule ~\ref{3R2}.
\end{lemma}

\begin{proof}
By contradiction, assume that $\PR$ includes Rule~\ref{3R2}. Note that  Rule~\ref{3R2} is only enabled on $r_1$ and $r_4$. Suppose that the scheduler activates both $r_1$ and $r_4$ at the same time. Let $ST'$ be the reached configuration once such robots move. The following rules are then possible:

             \begin{enumerate}[a.]
              \item \label{3R'1} 010(1)001 :: $\leftarrow$
              \item \label{3R'2} 010(1)001 :: $\rightarrow$
              \item \label{3R'3} 000(1)010 :: $\rightarrow$ 
             \end{enumerate} 
             
\begin{itemize}
\item \textbf{Case (a)}: Rule(\ref{3R'1}) is enabled. Assume that the scheduler activates both $r_2$ and $r_3$ at the same time. Once they move, all robots on the ring have the same view. $(i)$ If they execute $000(1)100$ :: $\rightarrow$ then suppose that the scheduler activates them all at the same time, then the configuration reached is undistinguishable from the previous one that contradicts Theorem~\ref{th:undist}. $(ii)$ if they execute $000(1)100$ :: $\leftarrow$. Assume that the scheduler activates only $r_2$ and $r_3$. Once the robots move the configuration reached is indistinguishable from $ST'$. A contradiction (Theorem~\ref{th:undist}).
\item \textbf{Case (b)}: Rule (\ref{3R'2}) is enabled. Assume that the scheduler activates both $r_2$ and $r_3$ at the same time. $r_2$ and $r_3$ become then neighbors. The following Rules are possible:
                      \begin{enumerate}[{b}1.]
                           \item \label{3R'11} 000(1)001 :: $\leftarrow$
                           \item \label{3R'22} 000(1)001 :: $\rightarrow$
                           \item \label{3R'33} 100(1)100 :: $\rightarrow$
                       \end{enumerate} 
                             
                         \begin{itemize}
                         \item \textbf{Case(b1)}. Rule (b\ref{3R'11}) is enabled. Suppose that the scheduler activates both $r_1$ and $r_4$ at the same time. Once the robots move they cannot see any occupied node (recall that $n>13$). Thus if they execute rule $000(1)000$ :: $\leftarrow \vee \rightarrow$, the robot move back to the previous location. A contradiction (Theorem~\ref{th:undist}).
                          \item \textbf{Case(b2)}.  Rule (b\ref{3R'22}) is enabled. Suppose that the scheduler activates both $r_1$ and $r_4$ at the same time. $(i)$ If $r_1$ and $r_4$ execute rule $000(1)011$ :: $\rightarrow$, a $1$.block is created. If the robots inside the $1$.block are the ones allowed to move, they either exchange their position (A contradiction, Theorem~\ref{th:undist}). Or two $2$-towers are created. Suppose that the scheduler always activate robots part of the $2$-tower at the same time. Thus, they will have the same behavior and act as a single robots. A contradiction (From \cite{DPT09c}: no exploration is possible with only $2$ robots even if the view is infinite).
$(ii)$ If rule $101(1)010$ :: $\rightarrow$ enabled then, if the scheduler activates both robots at the same time no exploration is possible since whatever the robots that move an indistinguishable configuration is reached. A contradiction (Theorem~\ref{th:undist}).    
                            \item \textbf{Case(b3)}. Rule (b\ref{3R'33}) is enabled. Suppose that the scheduler activates $r_2$ and $r_3$ at the same time. They will simply exchange their position. The configuration that is reached is undistinguishable from the previous one. A contradiction (Theorem~\ref{th:undist}).
                         \end{itemize}
                           \item \textbf{Case (c)}: Rule (\ref{3R'3}) is enabled. Once $r_1$ and $r_4$ move, two 1.blocks are created. If $r_2$ and $r_3$ are the ones allowed to move then if they execute $001(1)001$:: $\leftarrow$, two towers are created and no exploration is possible (refer to \cite{DPT09c}). If they execute $001(1)001$:: $\rightarrow$, a new $1$.block of size $2$ is created. Note that in the reached configuration $r_2$ and $r_3$ cannot move anymore. If $r_1$ and $r_4$ execute $000(1)011$ :: $\rightarrow$, a $1$.block of size $4$ is created and no exploration is possible (refer to Case (b)). If they execute $000(1)011$ :: $\leftarrow$, they become at distance $3$ from the $1$.block. Let refer to the reached configuration by $T'$. If $r_1$ and $r4$ keep being enabled on $T'$, they can only execute $000(1)001$ :: $\leftarrow$. Suppose that the scheduler activates both robots at the same time. Once they move they cannot see any other robot in $T'+1$. $(i)$ If  $000(1)000$ :: $\leftarrow$ $\vee$ $\rightarrow$ is enabled. Then suppose that the scheduler activates them at the same time such that they move back to their previous position. Thus, $T'+2$ is indistinguishable from $T'+1$. Contradiction (Theorem~\ref{th:undist}). $(ii)$ if $r_2$ and $r_3$ are the enabled robots then if they execute $100(1)100$ :: $\leftarrow$, then the configuration reached is indistinguishable from the configuration reached when Rule (\ref{3R'2}) is executed. Thus no exploration is possible in this case too.
\end{itemize}
\end{proof}

\begin{lemma}\label{L3R3}
Let $\PR$ be a semi-synchronous exploration protocol for $\phi = 3$, $n>13$, and $k<n$.  Then, $\PR$ does not include Rule~\ref{3R3}.
\end{lemma}

\begin{proof}

By contradiction, suppose that the scheduler activates both $r_2$ and $r_3$ at the same time. Once they execute Rule~\ref{3R3}, a $1$.block is created. Such that robots in the $1$.block cannot see any other robots. If $r_2$ and $r_3$ keep being enabled they can only move towards each other (otherwise the previous configuration is restored). Suppose that the scheduler activates them at the same time, the configuration reached is indistinguishable with the previous one. Contradiction (Theorem~\ref{th:undist}). If $r1$ and $r_4$ are the ones allowed to move then they can only execute $000(1)000$ :: $\leftarrow$ $\vee$ $\rightarrow$ (recall that $n>13$). Suppose that the scheduler activates them both at the same time such that they move in the opposite direction of the 1.block (let refer to the reached configuration by {\em IMP}). Since $000(1)000$ :: $\leftarrow$ $\vee$ $\rightarrow$ keeps being enabled on both $r1$ and $r_4$. Suppose that the scheduler activates them at the same time such that they move back to their previous position. The configuration reached is indistinguishable from {\em IMP}. Contradiction  (Theorem~\ref{th:undist}).   
\end{proof}

\begin{lemma}\label{IMPATOM}
No deterministic exploration is possible in the ATOM model for $\phi=3$, $n>13$ and $k<n$.
\end{lemma}

\begin{proof}
%
%

Follows from Lemmas \ref{L3R1}, \ref{L3R2} and \ref{L3R3}. 

\end{proof}   

\begin{lemma}
No deterministic exploration is possible in the CORDA model for $\phi=3$, $n>13$ and $k<n$.
\end{lemma}

\begin{proof}
Follows directly from Lemma \ref{IMPATOM}.
\end{proof}

\subsection{Two Asynchronous Algorithms} 
In the following, we present two deterministic algorithms that solve the exploration problem. The first one uses $7$ robots and works for any towerless initial configuration that contains a single $\phi$.group such that $n\geq k\phi+1$. The second one uses $5$ robots but works only when the starting configuration contains a $1$.block of size $5$. 

\subsubsection{Exploration using $k=7$.}

Before detailing our solution, let us first define some special configurations:

\begin{definition}
A configuration is called {\tt Set} at instant $t$ (refer to Figure \ref{fig:Set}) if there exists a sequence of $5$ nodes $u_i,u_{i+1}, \dots , u_{i+4}$ such that:
\begin{itemize}
\item $M_j=2$ for $j\in \{i+1, i+4\}$ 
\item $M_j=3$ for $j=i+1$ 
\item $M_j=0$ for $j\in \{i+2, i+3\}$ 
\end{itemize}
\end{definition}

\begin{definition}
A configuration is called {\tt Final} at instant $t$ (refer to Figure \ref{fig:Final}) if there exists a sequence of $4$ nodes $u_i,u_{i+1}, \dots , u_{i+3}$ such that:
\begin{itemize}
\item $M_j=2$ for $j\in \{i, i+1\}$ 
\item $M_j=0$ for $j=i+2$ 
\item $M_j=3$ for $j=i+3$ 
\end{itemize}
\end{definition}

\begin{figure}[H]
 \begin{minipage}[b]{.46\linewidth}
 \centering\epsfig{figure=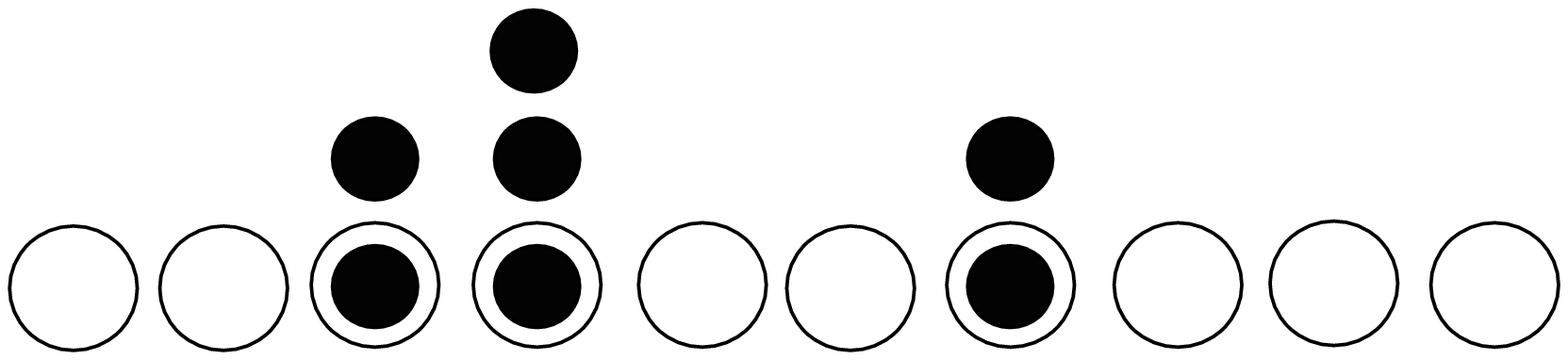,width=5cm}
 \caption{$Set$ configuration\label{fig:Set}}
 \end{minipage} \hfill
\begin{minipage}[b]{.46\linewidth}
  \centering\epsfig{figure=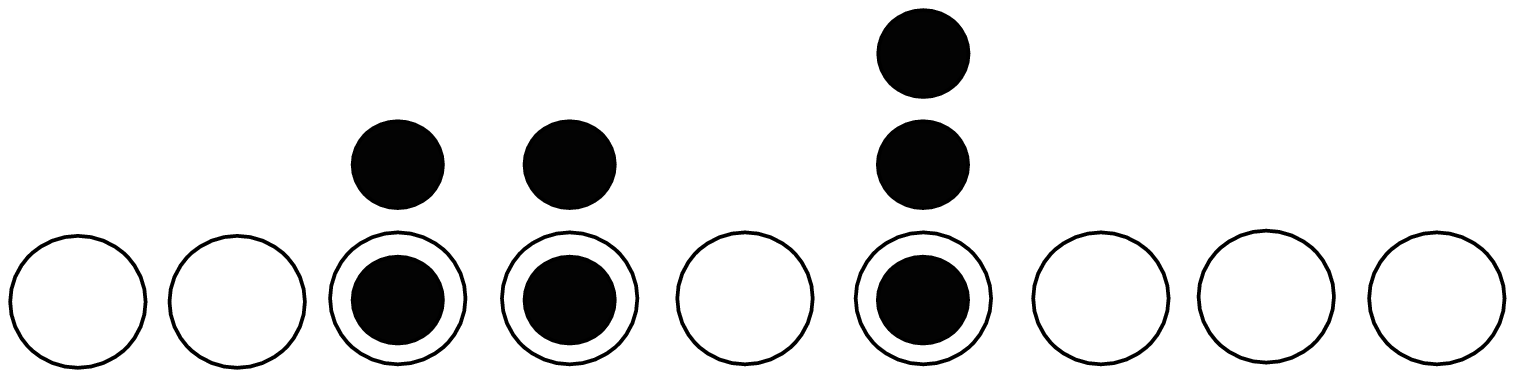,width=5cm}
 \caption{$Final$ configuration\label{fig:Final}}
\end{minipage} \hfill
\end{figure}

The algorithm comprises two phases as follow:

\begin{enumerate}
\item {\tt Set-Up Phase}. The aim of this phase is to create a {\tt Set Configuration}. The starting configuration contains a $\phi$.group, such that $n \geq k(\phi)+ 1$.
\item {\tt Exploration Phase}. The starting configuration of this phase is the {\tt Set Configuration}. A set of robots are elected to visit the ring's nodes. {\tt Final Configuration} is created at the end of this phase. 
\end{enumerate}

The formal Description of the algorithm is given in Algorithm \ref{algo:A31}.

\begin{algorithm}[htb] 
\caption{Asynchronous Exploration using 7 robots ($\phi=3$)}
\label{algo:A31}
\begin{scriptsize}
\begin{tabular}{rlcll}
\multicolumn{5}{l}{{\tt Set-Up} Phase}\\
3A1: & $000(1)001$ & $::$ & $\rightarrow$ & \comment{Move toward the robot at distance $3$}
\\
3A2: & $000(1)01?$ & $::$ & $\rightarrow$ & \comment{Move toward the robot at distance $2$} 
\\
3A3: & $000(1)1??$ & $::$ & $\rightarrow$ & \comment{Move toward my neighboring occupied node} 
\\
3A4: & $000(2)001$ & $::$ & $\rightarrow$ & \comment{Move toward the robot at distance $3$}
\\
3A5: & $000(2)01?$ & $::$ & $\rightarrow$ & \comment{Move toward the robot at distance $2$} 
\\
3A6: & $002(1)001$ & $::$ & $\rightarrow$ & \comment{Move in the opposite direction of the tower} 
\\
3A7: & $002(1)01?$ & $::$ & $\rightarrow$ & \comment{Move in the opposite direction of the tower} 
\\
3A8: & $021(1)001$ & $::$ & $\rightarrow$ & \comment{Move toward the robot at distance $3$}
\\
3A9: & $021(1)01?$ & $::$ & $\rightarrow$ & \comment{Move toward the robot at distance $2$}
\\
3A10: & $021(1)120$ & $::$ & $\leftarrow$ $\vee$ $\rightarrow$ & \comment{Move toward one of my neighboring node}
\\
3A11: & $220(1)200$ & $::$ & $\leftarrow$  & \comment{Move toward the tower at distance $2$}
\\
3A12: & $022(1)020$ & $::$ & $\leftarrow$  & \comment{Move toward the neighboring tower}
\medskip\\
\multicolumn{5}{l}{{\tt Exploration} Phase}\\
3A13: & $000(2)300$ & $::$ & $\leftarrow$ & \comment{Move in the opposite direction of the $3$.tower}
\\
3A14: & $001(1)300$ & $::$ & $\leftarrow$ & \comment{Move in the opposite direction of the $3$.tower}
\\
3A15: & $000(2)030$ & $::$ & $\leftarrow$ & \comment{Move in the opposite direction of the $3$.tower}
\\
3A16: & $001(1)030$ & $::$ & $\leftarrow$ & \comment{Move in the opposite direction of the $3$.tower}
\\
3A17: & $200(3)000$ & $::$ & $\leftarrow$ & \comment{Move toward the $2$.tower at distance $3$}
\\
3A18: & $201(2)00?$ & $::$ & $\leftarrow$ & \comment{Move to my neighboring occupied node}
\\
3A19: & $202(1)00?$ & $::$ & $\leftarrow$ & \comment{Move to my neighboring occupied node}
\\
3A20:& $200(2)030$ & $::$ & $\leftarrow$ & \comment{Move in the opposite direction of the $3$.tower}
\\
3A21:& $201(1)030$ & $::$ & $\leftarrow$ & \comment{Move towards the $2$.tower at distance $2$}
\\
3A22:& $000(2)020$ & $::$ & $\rightarrow$ & \comment{Move towards the $2$.tower at distance $2$}
\\
3A23:& $000(1)120$ & $::$ & $\rightarrow$ & \comment{Move towards my neighboring occupied node}
\end{tabular}
\end{scriptsize}
\end{algorithm}


\paragraph{\textbf{Proof of correctness.}} 
We prove in the following the correctness of Algorithm \ref{algo:A31}.

     
\begin{lemma}\label{Mtower}
During the first phase, if there is a rule that is enabled on robot part of a $2$.tower such that the scheduler activates only one robot in the $2$.tower (the tower is destroyed), then the tower is built again and the configuration reached is exactly the same as when both robots in the $2$.tower were activated at the same time by the scheduler. 
\end{lemma}     

\begin{proof}
The rules that can be executed by the robots in a $2$.tower in the first phase are Rules $3A4$ and $3A5$. In both Rules, robots in the tower can see only one robot at one side. Their destination is their adjacent empty node towards the robot that can be seen. Suppose that the scheduler activates only one robot in the $2$.tower. In the configuration reached Rule $3A3$ is enabled. Note that the robot on which one of this rule is enabled is the one that was in the $2$.tower and did not move. Once it execute $3A3$, it moves to its adjacent occupied node, thus a $2$.tower is created again and the configuration reached is exactly the same as when both robots were activated by the scheduler. Hence the lemma holds.
\end{proof}

\begin{lemma}\label{CUnited}
Starting from any towerless configuration that contains a single $\phi$.group, a configuration containing an
$S^2$-sequence is eventually reached.
\end{lemma}

\begin{proof}
Two cases are possible as follow:
\begin{enumerate}
\item The starting configuration contains a single 1.block of size $7$. Rule $3A3$ is then enabled on the two robots that are at the extremity of the 1.block. If the scheduler activates both robots that the same time, two $2$.towers are created and the lemma holds. If the scheduler activates only one robot, then Rule $3A3$ is the only rule enabled in the system. When the rule is executed, the robot that was supposed to move moves to its neighboring occupied node and a $2$.tower is created. Thus in this case too the lemma holds.

\item Other cases. Let consider only one extremity of the $\phi$.group. If the robot at this extremity (let this robot be $r_1$) does not have a neighboring occupied node, then either Rule $3A1$ or $3A2$ is enabled. Once one of these rules is executed on $r_1$, $r_1$ becomes closer to an occupied node. Thus it becomes neighbor of an occupied node in a finite time. Rule $3A3$ becomes then enabled on $r_1$. Once $r_1$ moves, a $2$.tower is created.
If the $2$.tower does not have any neighboring occupied node, either Rule $3A4$ or $3A5$ is enabled on robots in the $2$.tower. According to Lemma \ref{Mtower}, robots in the same $2$.tower move eventually as when they are activated at the same time. By moving they become closer to an occupied node. Hence they become eventually neighbors of an occupied node. $(i)$ If the robot that is neighbor of the $2$.tower has an empty node as a neighbor then it moves to its adjacent empty node (refer to Rule $3A6$). $(ii)$ in the case it has a neighboring occupied node besides the $2$.tower, then if the robot on its neighboring occupied node move in the opposite direction of the $2$.tower if it has a neighboring empty node (refer to Rules $3A7$ and $3A8$). By doing so the tower becomes neighbor of a 1.block of size $3$. Robots at the other extremity have the same behavior since the only rules that can be executed are the ones that are enabled at the extremity. Thus, a configuration containing an
$S^2$-sequence is reached in a finite time and the lemma holds.
\end{enumerate}
\end{proof}

\begin{lemma}\label{CSet}
Starting from a configuration containing an $S^2$-sequence, {\tt Set} configuration is eventually reached.
\end{lemma}

\begin{proof}
In a configuration containing an $S^2$-sequence, Rule $3A10$ is enabled. When the rule is executed, a new $2$.tower is created. Rule $3A11$ becomes then enabled on the robot not part of a $2$.tower (let this robot be $r_1$). Once it moves, only Rule $3A12$ becomes enabled on $r_1$. When the robot is activated a $3$.tower is created. {\tt Set} configuration is then reached and the lemma holds.  
\end{proof}

\begin{lemma}\label{CFinal}
Starting from a {\tt Set} configuration, {\tt Final} configuration is eventually reached.
\end{lemma}

\begin{proof}
Let $T1$ and $T2$ be the $2$.towers in the Set configuration such that $T1$ has a neighboring $3$.tower that we call $T3$. Rule $3A13$ is enabled on the robots that are part of $T1$. If the scheduler activates only one robot, Rule $3A14$ becomes enabled on the robot that was supposed to move, thus the configuration reached is exactly the same as when both robots in $T1$ were activated at the same time. Rule $3A15$ becomes enabled on the $T1$. Once the robots on $T1$ are activated, the tower becomes at distance $3$ from $T3$ (in the case the scheduler activates only one robot in $T1$, only Rule $3A16$ becomes enabled on the robot that was supposed to move, once the scheduler activates the robot, the configuration reached is exactly the same as when both were activated at the same time). Robots on $T3$ are the only ones allowed to move (refer to Rule $3A17$). Once they move, $T3$ becomes at distance $2$ from $T1$ (in the same manner, if the scheduler activates only some robots in $T3$ then either Rule $3A18$ or Rule $3A19$ is the only one enabled. Thus, $T3$ is built again (the configuration reached is similar to the one that was reached when all robots in the $T3$ have moved at the same time). Robots in $T1$ are now the only one allowed to move and so on. Observe that the distance between $T3$ and $T2$ increases while the distance between $T1$ and $T2$ decreases. When $T2$ and $T1$ becomes at distance $3$. Rule $3A20$ is then enabled on $T1$ (if the scheduler activates only one robot then Rule $3A21$ is enabled and hence $T1$ is built again). Robots in $T2$ are now the only ones enabled (refer to Rule $3A22$), when they move, {\tt Final} configuration is reached and the lemma holds (In the case the scheduler activates only one robot in $T2$, Rule $3A23$ becomes enabled on the robot that did not move. Once it is activated, {\tt Final} configuration is reached and the lemma holds).
\end{proof}

\begin{lemma}
Starting from a any towerless configuration that contains a single $\phi$.group, {\tt Final} configuration is eventually reached and all the nodes of the ring have been explored.
\end{lemma}

\begin{proof}
From Lemmas \ref{CUnited}, \ref{CSet} and \ref{CFinal} we can deduce that {\tt Final} configuration is eventually reached starting from a towerless configuration that contains a single $\phi$.group. In another hand, the node between $T1$ and $T2$ in the {\tt Set} configuration have been explored since
they were occupied when the configuration contained an $S^2$-sequence. When the set of robots in charge of performing the exploration task move on the ring, the distance between $T1$ and $T2$ decreases such that when {\tt Final} configuration is reached $T1$ and $T2$ become neighbor. Thus, we are sure that all the nodes of the ring has been explored and the lemma holds. 
\end{proof}

\subsubsection{Exploration using $k=5$.}

In the following, we provide two definitions that are needed in the description of our algorithm:

\begin{definition}
A configuration is called {\tt Set2} at instant $t$ (refer to Figure \ref{fig:Set2}) if there exists a sequence of $6$ nodes $u_i,u_{i+1}, \dots , u_{i+5}$ such that:
\begin{itemize}
\item $M_j=1$ for $j\in \{i, i+5\}$ 
\item $M_j=3$ for $j=i+1$ 
\item $M_j=0$ for $j\in \{i+2, i+3, i+4\}$ 
\end{itemize}
\end{definition}

\begin{definition}
A configuration is called {\tt Done} at instant $t$ (refer to Figure \ref{fig:Done}) if there exists a sequence of $4$ nodes $u_i,u_{i+1}, \dots , u_{i+3}$ such that:
\begin{itemize}
\item $M_j=1$ for $j\in \{i, i+1\}$ 
\item $M_j=0$ for $j=i+2$ 
\item $M_j=3$ for $j= i+3$ 
\end{itemize}
\end{definition}

\begin{figure}[H]
\begin{minipage}[b]{.46\linewidth}
  \centering\epsfig{figure=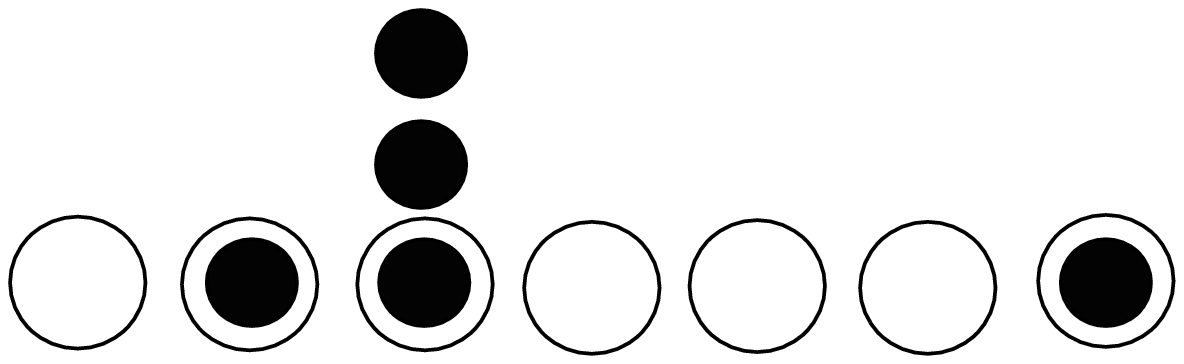,width=4cm}
 \caption{$Set2$ \label{fig:Set2}}
 \end{minipage} \hfill
\begin{minipage}[b]{.46\linewidth}
  \centering\epsfig{figure=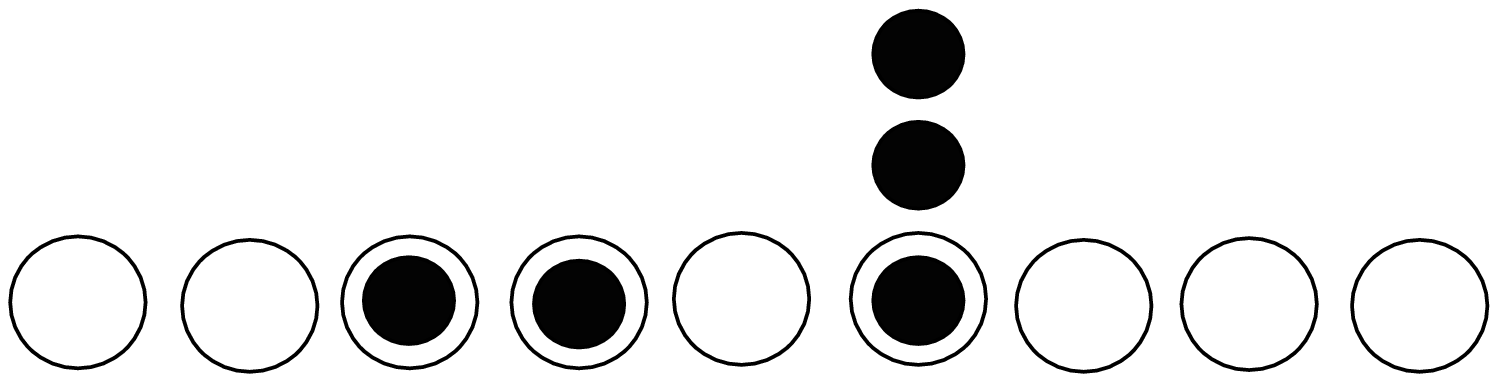,width=4cm}
 \caption{$Done$ configuration\label{fig:Done}}
\end{minipage}\hfill
\end{figure}  

The algorithm comprises also two phases as follow:

\begin{enumerate}
\item {\tt Set-Up Phase}. The initial configuration of this phase contains a $\phi$.group such that $n\geq k\phi+1$. The aim of this phase is to create a {\tt Set2} configuration.
\item {\tt Exploration Phase}. The starting configuration of this phase is {\tt Set2} Configuration. A set of robots are elected to perform the exploration task. At the end of this phase, {\tt Done} configuration is created.
\end{enumerate}

The formal description of our algorithm is given in Algorithm \ref{algo:53A}.

\begin{algorithm}[htp]
\caption{Asynchronous Exploration using $5$ robots ($\phi=3$)}
\label{algo:53A}
\begin{scriptsize}
\begin{tabular}{rlcll}
\multicolumn{5}{l}{{\tt Set-Up} Phase}\\
%
3A'1: & $011(1)110$ & $::$ & $\rightarrow$ $\vee$ $\leftarrow$ & \comment{Move to one of my neighboring node}
\\
3A'2: & $120(1)100$ & $::$ & $\leftarrow$ & \comment{Move to towards the tower at distance $2$}
\\
3A'3: & $012(1)010$ & $::$ & $\leftarrow$ & \comment{Move to the tower} 
\\
3A'4: & $300(1)000$ & $::$ & $\rightarrow$ &\comment{Move in the opposite direction of the $3$.tower}
\\
3A'5: & $120(1)101$ & $::$ & $\leftarrow$ &\comment{Move toward the $2$.tower (case $n=6$)}
\\
3A'6: & $300(1)001$ & $::$ & $\rightarrow$ &\comment{Move toward in the opposite direction of the $3$.tower (case $n=7$)}
\medskip\\
\multicolumn{5}{l}{{\tt Exploration} Phase}\\
3A'7: & $000(1)300$ & $::$ & $\leftarrow$ & \comment{Move in the opposite direction of the $3$.tower}
\\
3A'8: & $000(1)030$ & $::$ & $\leftarrow$ & \comment{Move in the opposite direction of the $3$.tower}
\\
3A'9: & $100(3)000$ & $::$ & $\leftarrow$ & \comment{Move toward the isolated robot}
\\
3A'10: & $101(2)000$ & $::$ & $\leftarrow$ & \comment{Move to my neighboring occupied node}
\\
3A'11: & $102(1)000$ & $::$ & $\leftarrow$ & \comment{Move to the $2$.tower}
\\
3A'12: & $100(1)030$ & $::$ & $\leftarrow$ & \comment{Move in the opposite direction of the $3$.tower}
\\
3A'13: & $010(1)030$ & $::$ & $\leftarrow$ & \comment{Move in the opposite direction of the $3$.tower}
\\
3A'14: & $010(1)300$ & $::$ & $\leftarrow$ & \comment{Move in the opposite direction of the $3$.tower (case $n=6$ $\wedge$ $n=7$ )}
\end{tabular}
\end{scriptsize}
\end{algorithm}

\paragraph{Proof of Correctness}
We prove in the following the correctness of our solution.

\begin{lemma}\label{lem:Set2}
Starting from a configuration that contains a $1$.block of size $5$ such that $n>k+2$, {\tt Set2} configuration is eventually reached.
\end{lemma}

\begin{proof}
When the configuration contains a 1.block of size $5$, Rule $3A'1$ is enabled. Once the rule is executed a $2$.tower is created. On the robot that is at distance $2$ from the $2$.tower (let this robot be $r_1$) Rule $3A'2$ is enabled. When the rule is executed $r_1$ becomes neighbor of the $2$.tower. Rule $3A'3$ becomes the only one enabled then enabled. When $r_1$ is activated, it executes the rule and it joins the $2$.tower and hence a $3$.tower is created. Rule $3A'4$ becomes the only rule enabled on the robot that is at distance $3$ from the $3$.block, when the robot is activated, it becomes at distance $4$ from the $3$.tower. The {\tt Set2} configuration is then created and the lemma holds.

\end{proof}

\begin{lemma}\label{lem:Done}
Starting from a configuration of type {\tt Set2}, {\tt Done} configuration is eventually reached.
\end{lemma}

\begin{proof}

Let $r_1$ be the robot that is neighbor of the $3$.tower and let $r_2$ be the robot that cannot see any other robot. $r_1$ is the only robot allowed to move and this until it becomes at distance $3$ from the tower (see Rules $3A'7$ and $3A'8$). Rule $3A'9$ becomes enabled on the robots part of the $3$.tower, their destination is their adjacent empty node towards $r_1$. In the case the scheduler activates some robots in the tower, the remaining robots (the one that were supposed to move) are the only ones allowed to move, their destination is their adjacent occupied node (see Rule $3A'10$ and $3A'11$). Thus the configuration reached is exactly the same as when all the robots that are in the $3$.tower are activated by the scheduler at the same time. Hence, the tower becomes eventually at distance $2$ from $r_1$. $r_1$ is now allowed to move, its destination is its adjacent empty node in the opposite direction of the $3$.tower. Note that once it moves it becomes at distance $3$ from this tower. Thus, robots in the $3$.tower are the ones allowed to move and so on. Both the $3$.tower and $r_1$ keep moving in the same direction such that at each time they move they become closer to $r_2$ (the distance between $r_1$ and $r_2$ decreases). Thus, $r_1$ and $r_2$ becomes eventually neighbors. Note that the robots in the $3$.tower cannot see $r_2$ yet, so it will continue to move towards $r_1$. When all of them move, $Done$ configuration is reached and the lemma holds.
\end{proof}

\begin{lemma}
Starting from a configuration that contains a 1.block, when {\tt Done} configuration is reached, and all the nodes of the ring have been explored.
\end{lemma}  

\begin{proof}
From Lemmas \ref{lem:Set2} and \ref{lem:Done} we deduce that staring from any towerless configuration that contains a single $1$.block such that $n>k+2$, {\tt Done} configuration is reached in a finite time. Let $T$ be the $3$.tower and let $r_1$ (resp $r_2$) be respectively the robots that is neighbor of the $3$.tower (resp the robot that cannot see any other robot). In {\tt Set2} configuration, the nodes between the $T$ and $r_2$ has been already visited since the starting configuration was a 1.block. $T$ and $r_1$ keep moving in the same direction such that the distance between $r_1$ and $r_2$ decreases at each time. When {\tt Done} configuration is reached, $r_1$ and $r_2$ become neighbor. We can then deduce that all the nodes of the ring have been explored and the lemma holds. 
\end{proof}

\section{Conclusion}
\label{sec:conclu}
In this paper, we studied the exploration of uniform rings by a team of oblivious robots. The assumptions
of unlimited visibility made in previous works has enabled them to focus only on overcoming the computational
weaknesses of robots introduced by the simultaneous presence of obliviousness and asynchrony in the design of
exploration algorithms.  In this paper, we added one more weakness: {\em Myopia} {\em ie,} robots have only a limited
visibility. We provided evidences that the exploration problem can still be solved under some conditions by
oblivious robots despite myopia.  We studied the problem for both synchronous and asynchronous settings,
and considered three types of visibility capabilities: $\phi=1$, $\phi = 2$, and $\phi =3$.  

The complete characterization for which the exploration of the ring by myopic robots is solvable remains open
in general. We conjuncture that the solutions proposed in this paper are optimal with respect to number of robots. We
also believe that the condition $n> k\phi +1$ is a necessary condition to solve the problem.  
Also, the problem of exploring other topologies and arbitrary graphs by myopic robots is a natural extension of this work.

\label{sec:concl}

\bibliographystyle{plain}
\bibliography{biblio-1}

\end{document}